\documentclass[a4paper,11pt]{article}
\usepackage[left=3.5cm, right=3.5cm, top=3cm, bottom=3.5cm]{geometry}
\usepackage{sfspaper}
\usepackage{texab}
\usepackage{leftidx}
\usepackage{eufrak}
\usepackage{graphicx}
\usepackage{natbib}
\usepackage{parskip}

\usepackage[ruled,vlined,linesnumbered]{algorithm2e}

\usepackage{amsmath, amsthm, amssymb}

\usepackage{subfigure} %-- subfigures : e.g. 2 fig, side by side

\newtheorem{lemma}{Lemma}[section]

\newtheorem{thm}{Theorem}[section]
\newtheorem{defi}{Definition}[section]
\newcommand{\Perp}{\perp \! \! \! \perp}

\begin{document}

\bibliographystyle{abbrv} 
%\citationstyle{dcu}

\Title{High dimensional sparse covariance estimation via directed acyclic
  graphs} 

\Author{Philipp R\"utimann and Peter B\"uhlmann}

\date{November 2009}
%%\Resrep{number}%-- AFTER \Date , BEFORE \maketitle

\maketitle % Macht einen Titel im LaTeX-Stil, der auch fuer Res.Rep.s
           % nuetzlich sein kann

\begin{abstract}
\normalsize
We present a graph-based technique for estimating sparse covariance
matrices and their inverses from high-dimensional data. The method
is based on learning a directed 
acyclic graph (DAG) and estimating parameters of a multivariate Gaussian
distribution based on a DAG. For inferring the
underlying DAG we use the PC-algorithm \cite{SpiPGS00}
and for estimating the DAG-based covariance matrix and its inverse, we 
use a 
Cholesky decomposition approach which provides a positive (semi-)definite
sparse estimate. We present a consistency result in the high-dimensional
framework and we compare our method with the Glasso \cite{FHT07,ABG08,BGA08} for
simulated and real data. 
\end{abstract}

\section{Introduction}

Estimation of covariance matrices is an important part of
multivariate analysis. There are many problems with high-dimensional
data where an estimation of the covariance matrix is of interest, for
example in principal component analysis or classification by discriminant
analysis. Application areas where such problems arise include gene
micro-arrays, imaging and image classification or text retrieval. In many of
these applications, the primary goal is the estimation of the inverse of a
covariance matrix $\Sigma^{-1}$, also known as the precision or
concentration matrix, rather than the covariance $\Sigma$ itself. In
low-dimensional settings with $p<n$, where $p$ 
denotes the row- or column-dimension of $\Sigma$ and $n$ the sample size, we can
obtain an estimate of $\Sigma^{-1}$ by estimation and inversion of
the Gaussian maximum likelihood estimator $\hat{\Sigma}_{MLE}$. But when p is
large, inversion of this estimate is problematic and its accuracy is very poor. 

Recently, two classes for high-dimensional covariance estimation have
emerged: those that rely on a 
natural ordering among variables and typically assuming that variables far
apart in 
the ordering are only weakly correlated, and those which are invariant to
variable permutation. Regularized estimation by
banding or tapering \cite{BL04,FB07,BL08} or using sparse
Cholesky factors of the inverse covariance matrix relying on the natural
ordering of the variables \cite{WP03,HPL06,LRZ08}  are members of the first
class of covariance estimators. When having no natural ordering among the
variables, estimators should be permutation invariant with respect to
indexing the variables. A popular approach to obtain a sparse
permutation-invariant estimate is to add a Lasso penalty on the entries of
the concentration matrix to the negative Gaussian log-likelihood
\cite{FHT07,ABG08,BGA08,RBLZ08}. This amounts to shrinking some of the
elements of the inverse covariance matrix exactly to zero. Alternatively,
the Lasso can be used for inferring an undirected conditional independence
graph using node-wise regressions \cite{MeinB06} and a covariance estimate
can then be obtained using the structure of the graph. Other approaches
include a simple hard-thresholding of the elements of the
unpenalized maximum likelihood estimator \cite{BLthresh08}, with the
disadvantage that the resulting estimate is not necessarily positive
(semi-) definite.  

The method which we present here is also invariant under permutation of
the variables. The type of regularization which we pursue is based on
exploiting a sparse graphical model structure first and then estimating
the covariance matrix and its inverse using non-regularized
estimation. Because of the sparsity of the graphical model structure, the 
second step does not need any regularization anymore. 
More precisely, we use a sparsely structured Cholesky decomposition
of the concentration matrix for estimation of the covariance and
concentration matrix. To obtain the structure of such a Cholesky factor, we
estimate a DAG (in fact, an equivalence class of DAGs). Thus, this approach
enforces a completely different sparsity structure on the Cholesky factor
than proposals for ordered data as in e.g. \cite{BL04,FB07,BL08,DY09}.

For a given DAG, our approach equals the iterative conditional fitting
(ICF) method presented in \cite{DR04,ChDR07} which reduces here to the
standard technique of fitting Gaussian DAG models. Our contribution is to
use an estimated DAG, i.e. an estimated equivalence class of DAGs from the
PC-algorithm \cite{SpiPGS00}, and to
analyze the method in the high-dimensional case taking the uncertainty of
structure estimation of the equivalence class of DAGs into account. 
We argue in this paper that within the class of methods which are invariant
under variable permutation, a graph-structured approach can be worthwhile
for a range of scenarios, sometimes resulting in performance gains up to
30-50\% over shrinkage methods. 

In Section \ref{graphterm} we give a brief overview over graph terminology
and graphical models. Section \ref{meth} introduces our methodology and we
show asymptotic consistency of the method in the high-dimensional framework
in Section \ref{consi}. Simulations and real data examples are presented in
Section \ref{simu} and we propose a robustified version of our procedure in
Section \ref{robu}.   

\section{Graph terminology and graphical models}\label{graphterm}

\subsection{Graphs}

Let $G=(V,E)$ be a graph with a set of vertices $V$ and a set of
edges $E\subseteq V\times V$. In our context, we use $V = \{1,\ldots ,p\}$
corresponding to some random variables $X_1,...,X_p$. 

A graph can be \emph{directed}, \emph{undirected} or \emph{partially
  directed}. An edge between two vertices, for example $i$ and $j$, is 
called \emph{directed} if the edge has an arrowhead: $i\leftarrow j$ or
$i\rightarrow j$. An edge without arrowhead is an \emph{undirected} edge:
$i- j$. A graph in which all edges are directed is called a \emph{directed
graph}; and vice-versa, a graph in which no edge is directed is called
an \emph{undirected graph}. A graph which may contain both directed and 
undirected edges is called a \emph{partially directed graph}. The
underlying undirected graph of a (partially) directed graph G which we
derive by removing all the arrowheads is called the 
\emph{skeleton} of G.\\ 
Two vertices $i$ and $j$ are \emph{adjacent} if there is any kind of
edge between them. The \emph{adjacency set} of a vertex $i$, denoted by
$adj(i,G)$, is the set of all vertices that are adjacent to $i$ in G. A
\emph{path} is a sequence of vertices $\left\{1,\ldots,k\right\}$ such
that $i$ is adjacent to $i+1$ for each $i=1,\ldots,k-1$. A \emph{directed
path} is a path with directed edges that follows the direction of the
arrows. When the first and last vertices coincide, the directed path is
called a \emph{directed cycle}. An \emph{acyclic} graph is a graph that contains no directed cycles. A directed graph with no directed cycles is called
\emph{directed acyclic graph} (DAG).\\  
If $i\rightarrow j$, then $i$ is
called a \emph{parent} of $j$ and $j$ is called a \emph{child} of
$i$. The set of parents of $i$ in G is denoted as $pa(i)$ and the set of
children as $ch(i)$. If there is a directed path from $i$ to $j$, then
$i$ is called an \emph{ancestor} of $j$ and $j$ is called an
\emph{descendant} of $i$. The set of ancestors of $i$ is denoted
as $an(i)$, the set of descendants as $de(i)$ and the set of
non-descendants as $nde(i)$. A \emph{v-structure} in a
graph G is an  ordered triple of vertices $(i,j,r)$, such that
$i\rightarrow j$ and $j\leftarrow r$, and $i$ and $r$ are not
adjacent in G. The vertex $j$ is then called a \emph{collider}.\\
A path $Q$ from $i$ to $j$ in a directed acyclic graph G is said to be
\emph{blocked} by S, if it contains a vertex $v\in Q$ such that either
(i) $v \in S$ and $v$ is no collider; or (ii) $v \notin S$ nor has $v$ any
descendants in S, and $v$ is a collider. A path that is not blocked by S is
said to be \emph{active}. Two subsets A and 
B are said to be \emph{d-separated} by S if all paths from A to B are
blocked by S. In other words, there is no active path from A to B. 

\subsection{Graphical models and Markov properties}\label{markov} 

Graphical models form a probabilistic tool to analyze and visualize
conditional dependence between random variables, using some encoding with
edges in the graph. Fundamental to the idea of a graphical model is,
based on graph theoretical concepts and algorithms, the notion of
modularity where a complex system is built by combining simpler parts.
One can distinguish between three main graphical models.    
Here we focus on DAG models, where all the edges of the graph are directed.
According to \cite{LauS96}, a DAG model may exhibit several directed Markov
properties. In the following, we present only two of them.\\
We use the following notation. Let $P$ denote the distribution of
$(X_1,...,X_p)$. For $x \in \R^p$, we denote by $x_A = \{x_j;\ j \in A\}$
for $A \subseteq V = \{1,\ldots ,p\}$ and analogously for the random
vector $X_A$. Furthermore, for disjoint subsets $A,B$ and $S$, we denote by
$X_A \Perp X_B|X_S$ conditional independence between $X_A, X_B$ given $X_S$.  

\begin{defi}\label{defi1}[Directed global Markov property]

Let $A,B$ and $S$ be disjoint subsets of $V$ and $G$ a DAG on $V$ . If  
\[
X_A \Perp X_B|X_S
\]
whenever A and B are d-separated by S in the graph $G$, we say P obeys the
directed global Markov property relative to the DAG $G$.   
\end{defi} 

\begin{defi}\label{defi2}[Recursive factorization property]
\label{fac}

We say that P admits a recursive factorization according to a DAG G whenever
there exist non-negative functions $f_i(.|.)\ (i=1,\ldots ,p)$, such that  

\[
\int f_i(x_i|x_{pa(i)})\nu(dx_i)=1
\]

and P has a density $f$ with respect to the measure $\nu$, where

\[
f(X_1,...,X_p)=\prod_{i=1}^pf_i(X_i|X_{pa(i)}).
\]

\end{defi}

If the density $f$ of $P$ is strictly positive, as for example in the case of
a multivariate Gaussian distribution, both Markov
properties in Definitions \ref{defi1} and \ref{defi2} are equivalent. For
more details see \cite[pp.46-52]{LauS96}.

A DAG model encodes conditional independence relationships via the notion
of d-separation. Several DAGs could encode the same set of conditional
independence relationships. These DAGs form an equivalence class,
consisting of 
DAGs with the same skeleton and v-structures. A \emph{complete partially
  directed acyclic graph} (CPDAG) uniquely describes such an equivalence
class. In fact, directed edges in the CPDAG are common to all DAGs
in the equivalence class. Undirected edges in the CPDAG correspond to edges that
are directed one way in some DAGs and another way in other DAGs of the
equivalence class. The absence of edges in the CPDAG means that
all DAG members in the equivalence class have no corresponding edge. 
If all the conditional independence relationships of a
distribution $P$ and no additional conditional independence relations, can be inferred from the graph, we say that the distribution
$P$ is \emph{faithful} to the DAG G. More precisely, if $P$ is faithful to
the DAG $G$: for any triple of disjoint sets $A,B$ and $S$ in $V$, 
\begin{eqnarray*}
 X_A \Perp X_B|X_S \Leftrightarrow A\ \mbox{and}\ B\ \mbox{are
   $d$-separated by}\ S\ \mbox{in}\ G. 
\end{eqnarray*}
Note that the directed global Markov property in Definition \ref{defi1} implies the
implication from the right- to the left-hand side; the other direction is
due to the faithfulness assumption. 
   
\section{Covariance estimation based on DAGs}
\label{meth}

Our methodology is based on two steps. We first infer the
CPDAG, i.e. the equivalence class of DAGs, and we then estimate the
covariance (concentration) matrix based on the CPDAG structure.

We assume throughout the paper that the data are
\begin{eqnarray}\label{data}
& & X^{(r)} = (X^{(r)}_{1},\ldots ,X^{(r)}_{p}),\ r=1,\ldots,n\nonumber\\
& & X^{(1)},\ldots, X^{(n)}\ \mbox{i.i.d.}\ \sim P
\end{eqnarray}
with $P$ being multivariate normal $\mathcal{N}_p(0,\Sigma)$, Markovian (as
in Definition \ref{defi1} or \ref{defi2}) and faithful to a DAG G. 

The Gaussian assumption implies that $\ERW{X_i\mid X_{pa(i)}}$ is
linear in $X_{pa(i)}$ which will be useful in the second estimation
step for the concentration or covariance matrix. Moreover, it allows us to
equate conditional independence with zero  
partial correlation which makes estimation for the CPDAG much easier. 

\subsection{Estimating the covariance matrix from a DAG}\label{estfromDAG}

We first assume that the underlying DAG is given. Using the factorization
property from Definition \ref{fac} in Section \ref{markov} we have:
\begin{equation*}
  f(X_1,...,X_p)=\prod_{i=1}^pf(X_i|X_{pa(i)}).
\end{equation*}
We use here and in the sequel the short-hand notation $f(\cdot|\cdot)$
instead of $f_i(\cdot|\cdot)$. 
For data as in (\ref{data}), we can then write the likelihood function as
\begin{equation*}
  L=\prod_{r=1}^nf(X^{(r)}_{1},...,X^{(r)}_{p})
  =\prod_{r=1}^n\prod_{i=1}^pf(X^{(r)}_{i}|X^{(r)}_{pa(i)})
  =\prod_{i=1}^p\prod_{r=1}^nf(X^{(r)}_{i}|X^{(r)}_{pa(i)}).  
\end{equation*}

Using the Gaussian assumption this leads to
the likelihood in terms of the unknown parameter $\Sigma$ (or $\Sigma^{-1}$
respectively).   

\begin{equation*}
  L(\Sigma)=\prod_{i=1}^pL_{i}(\mu_{i|pa(i)},\Sigma_{i|pa(i)}) 
\end{equation*}   

where $\mu_{i|pa(i)}$ and $\Sigma_{i|pa(i)}$ are the conditional
expectation and variance of $X_{i}$ given the
parents $X_{pa(i)}$. Note that the conditional covariance is a fixed
quantity whereas the conditional mean depends on the variables
$X_{pa(i)}$. For a single random variable $X_i$ we have:
\begin{equation}
  \label{bederw}
\begin{split}
  \mu_{i|pa(i)}=\ERW{X_i\mid X_{pa(i)}}&= \mu_i +
  \Sigma_{i,pa(i)}(\Sigma_{pa(i),pa(i)})^{-1}(X_{pa(i)}-\mu_{pa(i)})\\
&=\Sigma_{i,pa(i)}(\Sigma_{pa(i),pa(i)})^{-1}X_{pa(i)},
\end{split}
\end{equation}

with assumption $\mu_i = 0\ \forall i$ from above, and:
\begin{equation}
\label{condcov}
  \Sigma_{i|pa(i)}=\Sigma_{i,i}-\Sigma_{i,pa(i)}(\Sigma_{pa(i),pa(i)})^{-1}\Sigma_{pa(i),i}. 
\end{equation}   
The expressions $\Sigma_{i,pa(i)}$ and
$\Sigma_{pa(i),pa(i)}$ are sub-matrices formed by selecting the
corresponding rows and columns from the full covariance matrix
$\Sigma$. For example, $\Sigma_{i,pa(i)}$ is the sub-matrix (or vector) of
$\Sigma$ with row $i$ and columns $j \in pa(i)$. The values 
$\mu_{i|pa(i)}$ and $\Sigma_{i|pa(i)}$, in the $i$th factor $L_i$ of the
likelihood, are connected to regression, as described next.

Consider for each node $i$ a regression from
$X_i$ on $X_{pa(i)}$, where $X_i|X_{pa(i)} \sim
\mathcal{N}(\mu_{i|pa(i)},\Sigma_{i|pa(i)})$. We can represent these $p$
regressions in matrix notation as follows:
\begin{equation}
  \label{amatform}
  A\begin{bmatrix}X_1\\ \vdots\\ X_p\\ \end{bmatrix}=\epsilon    
\end{equation}   
where A is a $p\times p$ matrix corresponding to the regressions
and $\epsilon$ is the vector of the error terms. That is:
\begin{eqnarray}\label{mata}
A_{ij} = \left\{ \begin{array}{ll}
-(\Sigma_{i,pa(i)}(\Sigma_{pa(i),pa(i)})^{-1})_j &\
    \mbox{if}\ j \in pa(i)\nonumber\\
1 &\ \mbox{if}\ j=i\nonumber\\
0 &\ \mbox{otherwise} 
\end{array}.
\right.
\end{eqnarray}
Now we can easily compute $\Sigma$ or $\Sigma^{-1}$,
because we can write (\ref{amatform}) as 
\begin{equation*}
(X_1,\ldots ,X_p)^T = A^{-1}\epsilon.     
\end{equation*}  
Hence,
\begin{eqnarray*}
\Sigma = \COV{(X_1,\ldots ,X_p)^T} = \COV{A^{-1}\epsilon} =
A^{-1}\COV{\epsilon}(A^{-1})^T,
\end{eqnarray*}
where
\begin{eqnarray}\label{coveps}
\COV{\epsilon}=\begin{bmatrix}

 \Sigma_{1|pa(1)}      &  & 0\\
  & \ddots & \\
 0      &  & \Sigma_{p|pa(p)}
\end{bmatrix}.
\end{eqnarray}
We then also easily obtain
\[
\begin{split}
  \Sigma^{-1} = (A^{-1}\COV{\epsilon}(A^{-1})^T)^{-1}
  = A^T\COV{\epsilon}^{-1}A.
\end{split}  
\]
See \cite[Chap. 3]{COXW96} and \cite{Marchetti06}  for more details.

Since $\Sigma$ and $\Sigma^{-1}$ are based on the structure of the DAG G
(via the matrix A) we write $\Sigma_G$ and $\Sigma_G^{-1}$. An
estimator is now constructed as
follows. Consider the maximum likelihood 
estimator 
\begin{eqnarray}\label{MLEcov}
\hat{\Sigma}^{MLE} = n^{-1} \sum_{j=1}^n (X^{(j)} -
\overline{X})(X^{(j)} - \overline{X})^T
\end{eqnarray}
 as an ``initial'' estimator $\hat{\Sigma}_{init}$ of $\Sigma$ and use the
 following plug-in estimators:
\begin{eqnarray}\label{estcov}
& &\hat{\Sigma}_G = \hat{A}^{-1}\COVH{\epsilon}(\hat{A}^{-1})^T,\nonumber\\
& &\hat{\Sigma}_G^{-1} = \hat{A}^T\COVH{\epsilon}^{-1}\hat{A},
\end{eqnarray}
where $\hat{A}$ and $\COVH{\epsilon}$ are as in
(\ref{coveps}) but with the plug-in estimates
$\hat{\Sigma}^{MLE}_{i,pa(i)}$, $(\hat{\Sigma}^{MLE}_{pa(i),pa(i)})^{-1}$
(for $\hat{A}$) 
and $\hat{\Sigma}^{MLE}_{i|pa(i)}$ (for $\COVH{\epsilon}$) using formula
(\ref{condcov}). 

Note that the estimators in (\ref{estcov}) are automatically
positive semi-definite having eigenvalues $\ge 0$ (and positive definite
assuming $\hat{\Sigma}_{i|pa(i)}>0$ for all $j$, which would fail only in
very pathological cases). Furthermore, we
could use another ``initial'' estimator than $\hat{\Sigma}^{MLE}$ for
estimating $\Sigma_{i,pa(i)}$, $\Sigma_{pa(i),i}$ and
$\Sigma_{pa(i),pa(i)}$. We are exploiting this possibility for a robustified
version, as discussed in Section \ref{robu}. Finally, the estimator in
(\ref{estcov}) is implemented in the R-package \texttt{ggm} \cite{ChDR07}.

\subsection{Inferring a directed acyclic graph}

The conditional dependencies between $X_1,...,X_p$ and hence the DAG are
usually not known. We use the PC-algorithm \cite{SpiPGS00} with estimated
conditional dependencies to infer the corresponding CPDAG G, i.e. the
equivalence class of DAGs (inferring the true DAG itself is well-known
to be impossible due to identifiability problems). 

Estimation of the skeleton and partial orientation of edges are the two
major parts of inferring a CPDAG. 
In the following we will describe these two steps. 

\subsubsection{Estimating the CPDAG}\label{subsec.estCPDAG}

In a first step, we start from a complete undirected graph. When two
variables $X_i$ $X_j$ are found to be conditional independent given $X_K$
for some set $K$, the edge $i-j$ is deleted: details are given in Algorithm
\ref{pcskel}. In a second step, the edges are oriented using the
conditioning sets $K$ which made edges drop out in the first step: details
are given in Algorithm \ref{orient}.\\ 
In the first step of the PC-algorithm, we need to
estimate the conditional independence relations between
$X_1,...,X_p$. Under the Gaussian assumption conditional independencies can be 
inferred from partial correlations. Then, the conditional independence
of $X_i$ and $X_j$ given $X_K = \{X_r;r\in K\}$, where $K \subseteq
\{1,...,p\}\backslash \{i,j\}$, is equivalent to the following: the partial
correlation of $X_i$ and $X_j$ given $\{X_r;r\in K\}$, denoted by
$\rho_{i,j|K}$, is equal to zero. This is an elementary  
property of the multivariate normal distribution, see
\cite[Prop. 5.2]{LauS96}. Hence to obtain estimates of conditional
independencies we can use estimated partial correlations
$\hat{\rho}_{i,j|K}$. For testing whether an estimated partial 
correlation is zero or not, we apply Fisher's z-transform
\begin{equation*}
  Z(i,j\mid K)=\frac{1}{2}\log \left(\frac{1+\hat{\rho}_{ij|K}}{1-\hat{\rho}_{i,j|K}}\right).
\end{equation*} 
 
Since $Z(i,j\mid K)$ has a $\mathcal{N}(0,(n-|K|-3)^{-1})$ distribution if
$\rho_{i,j|K}=0$ \cite{AndT84}, we have evidence that
$\rho_{i,j|K}\neq 0$ if 
\begin{equation*}
  \sqrt{n-|K|-3}|Z(i,j\mid K)|>\Phi^{-1}(1-\frac{\alpha}{2}),
\end{equation*} 

where $\Phi$ is the cumulative distribution function of the standard Normal
distribution and the significance level $0<\alpha <1$ is a tuning
(threshold) parameter of the PC-algorithm described in Algorithms
\ref{pcskel} and \ref{orient}. 

\begin{algorithm}[h]
\caption{The PC-algorithm for the skeleton}
\label{pcskel}
\KwIn{z-transform of estimated partial correlations, tuning parameter $\alpha$}
\KwOut{Skeleton of CPDAG G, separation sets S (used later for directing the
  skeleton)} 
Form the complete undirected graph $\tilde{G}$ on the set
$\{1,\ldots ,p\}$\;
$l=-1$; $G=\tilde{G}$\;
\Repeat{for each ordered pair of adjacent nodes $i$,$j$: $|adj(i,G)\backslash\{j\}|<l$}{
  $l=l+1$\;
  \Repeat{all ordered pairs of adjacent variables $i$ and $j$, such that
    $|adj(i,G)\backslash\{j\}|\geq l$ and $K\subseteq
    adj(i,G)\backslash\{j\}$ with $|K|=l$, have been tested for conditional
    independence}{
    Select an ordered pair of adjacent variables $i$, $j$ in G such that $|adj(i,G)\backslash\{j\}|\geq l$\; 
    \Repeat{edge $i$, $j$ is deleted or all $K\subseteq
      adj(i,G)\backslash\{j\}$ with $|K|=l$ have been chosen}{
      Choose $K\subseteq adj(i,G)\backslash\{j\}$ with $|K|=l$\;
      \If{$\sqrt{n-|K|-3}|Z(i,j\mid K)| \leq \Phi^{-1}(1-\alpha/2)$}{
        Delete edge $i$, $j$\;
        Denote this new graph by G\;
        Save $K$ in $S(i,j)$ and $S(j,i)$\;      
      }
    }
  }
}
\end{algorithm}

If $\rho_{i,j|K}=0$ is plausible, the edge $i-j$ is deleted and K is saved
in $S(i,j)$. We call $S=\{S(i,j); i,j\in \{1,\dots ,p\},\ i\neq j\}$
the separation sets. These sets are important for extending the 
estimated skeleton to a CPDAG as described below in Algorithm \ref{orient}.
 
\begin{algorithm}[h]
\caption{The PC-algorithm: extending the skeleton to a CPDAG}
\label{orient}
\KwIn{Skeleton G of CPDAG, separation sets $S$}
\KwOut{CPDAG}
\ForAll{pairs of nonadjacent variables $i$, $j$ with common neighbor
  $k$}{
  \If {$k \notin S(i,j)$}{
    Replace $i - k - j$ in Skeleton of G by $i \rightarrow k
    \leftarrow j$\;
  }
}

\Repeat{no more orienting of undirected edges is possible by the rules
  $\mathbf{R1}$ to $\mathbf{R3}$}{
$\mathbf{R1}$ Orient $j-k$ into $j\rightarrow k$ whenever there is an arrow 
$i\rightarrow j$ such that $i$ and $k$ are nonadjacent\;
$\mathbf{R2}$ Orient $i-j$ into $i\rightarrow j$ whenever there is a chain $i\rightarrow k \rightarrow j$\;
$\mathbf{R3}$ Orient $i-j$ into $i\rightarrow j$ whenever there are two chains
$i\rightarrow k \rightarrow j$ and $i\rightarrow l \rightarrow j$ such that
$k$ and $l$ are nonadjacent\;
% $\mathbf{R4}$ Order $i-j$ into $i\\rightarrow j$ whenever there are two
% chains $i-k \rightarrow l$ and $k\rightarrow l \rightarrow j$ such that
% $k$ and $l$ are nonadjacent.\;
}
\end{algorithm}
\cite{Meek95} showed that the rules in Algorithm \ref{orient} are sufficient
to orient all arrows in the CPDAG, see also \cite[pp.50]{PeaJ08}. The
PC-algorithm, described in Algorithms \ref{pcskel} and \ref{orient}, yields
an estimate $\hat{G}_{\mathrm{CPDAG}}(\alpha)$ of the true underlying CPDAG which
depends on the tuning parameter $\alpha$. 

\subsubsection{The PC-DAG covariance estimator} 

Having an estimate $\hat{G}_{\mathrm{CPDAG}}(\alpha)$ of the CPDAG, we pick
any DAG $\hat{G}_{\mathrm{DAG}}(\alpha)$ in the equivalence class of the
CPDAG. This can be done by directing undirected edges in the CPDAG at
random without creating additional v-structures or cycles. The estimate for
the covariance and concentration matrix is then: 
\begin{eqnarray}\label{DAGest-cov}
\hat{\Sigma}_{\hat{G}_{\mathrm{DAG}}(\alpha)},\
\hat{\Sigma}^{-1}_{\hat{G}_{\mathrm{DAG}}(\alpha)}\ \mbox{as in formula
  (\ref{estcov})},
\end{eqnarray}
and since the PC-algorithm for DAGs is involved, we call it the PC-DAG
covariance estimator. Its only tuning parameter is $\alpha$ used in the
PC-algorithm. As described in Section \ref{subsec.estCPDAG}, it has the
interpretation of a significance level for 
a single test whether a partial correlation is zero or not. The choice of
this tuning parameter $\alpha$ can be done using cross-validation of  the
negative out-of-sample log-likelihood. 

We remark that the zeros in
$\hat{\Sigma}^{-1}_{\hat{G}_{\mathrm{DAG}}(\alpha)}$ are the same for any
choice of a DAG in the estimated CPDAG
$\hat{G}_{\mathrm{CPDAG}}(\alpha)$. However, the non-zero estimated
elements of the estimated matrices will be slightly
different. To avoid an unusual random realization when selecting a DAG
from $\hat{G}_{\mathrm{CPDAG}}(\alpha)$, we can sample many DAGs and
average the corresponding estimates for $\Sigma^{-1}$ or $\Sigma$. 

In some cases, we need some small modifications of the P{C-DAG covariance
  estimator which are described in Appendix \ref{sec.modif}. Estimation of
  a CPDAG as described in Algorithm \ref{pcskel} and \ref{orient} is
  efficiently implemented in the \texttt{R}-package \texttt{pcalg}, as
  described in its reference manual \cite{pcalgMAN}. 

\section{Consistency}
\label{consi}

We prove asymptotic consistency of the estimation method in
high-dimensional settings where the number of variables $p$ can be
much larger than the sample size $n$. In such a framework, the model
depends on $n$ and this is reflected notationally by using the subscript
$n$. We assume:
 
\begin{itemize}

\item[(A)] The data is as in (\ref{data}) with distribution $P_n$ of $(X_{1},...,X_{p_n})$ being multivariate
  normal $\mathcal{N}(0,\Sigma_n)$, Markovian as in Definition \ref{defi1}
  or \ref{defi2} and faithful to a DAG $G_n$. 

\item[(B)] The variances satisfy: $\VAR{X_i} = \sigma^2_{n;i} \le \sigma^2 <
  \infty$ for all $i=1,\ldots ,p_n$. 

\item[(C)] The dimension $p_n= O(n^a)$ for some $0 \le a < \infty$.

\item[(D)] The maximal cardinality
  $q_n=\max_{i=1,...,p_n}\left|adj(i,G_n)\right|$ of the adjacency sets in
  $G_n$ satisfies  $q_n= O(n^{\frac{1}{2}-b})$ for some $0<b\le 1/2$.

\item[(E)] For any $i,j\in {1,...,p_n}$, let $\rho_{n;i,j|S}$ denote the
  partial correlation between $X_{i}$ and $X_{j}$ given $S$, where $S\in
  \{1,\ldots ,p_n\}\setminus{\{i,j\}}$. These partial correlations
  are bounded above and below:
  \begin{equation*}
    \sup_{n,i\ne j,S}\left|\rho_{n;i,j|S}\right|\le M 
  \end{equation*}
for some $M<1$, and 
\begin{equation*}
    \inf_{i,j,S}\left\lbrace\left|\rho_{n;i,j|S}\right|; \rho_{n;i,j|S}\ne
      0\right\rbrace\geq c_n 
  \end{equation*}
with $c_n^{-1}= O(n^d)$ for some $0<d<1/4 + b/2$, where $b$ is as in (D).

\item[(F)] For every DAG in the equivalence class of the true underlying
  CPDAG (induced by the distribution in assumption (A)), the conditional
  variances satisfy the following bound:\\
\begin{equation*}
\begin{split}
&\inf_{1 \le i \le p_n,\ j \in pa(i)} \VAR{X_{j}\mid X_{pa(i)\setminus
    j}} \geq r > 0,\\ 
&\inf_{1 \le i \le p_n} \VAR{X_{i}\mid X_{pa(i)}} \geq r > 0.
\end{split} 
\end{equation*}  
\end{itemize}

Assumption (C) allows the number of variables $p_n$ to grow as an arbitrary
polynomial in the sample size and reflects the high-dimensional 
setting. Assumption (D) is a sparseness assumption, requiring that the
maximal number of neighbors per node grows at a slower rate than
$O(n^{\frac{1}{2}})$. Assumption (F) is a regularity condition on
the conditional variances. Assumption (E), in particular the second part,
is a restriction which corresponds to the detectability of non-zero
partial correlations: obviously, we cannot consistently detect non-zero
partial correlations of smaller order than $\frac{1}{\sqrt{n}}$. For
sparse graphs with $b$ close to $1/2$ in (D), the value $d$ close to $1/2$
is allowed. i.e. close to the $1/\sqrt{n}$ detection limit. Under
assumptions (A)-(E), the PC-algorithm was shown to be consistent for
inferring the true underlying CPDAG \cite[Th.2]{DAGPC07}.  
More precisely, we denote by $\hat{G}_{\mathrm{CPDAG};n}(\alpha)$ the
estimate for the underlying CPDAG, using the PC-algorithm with tuning
parameter $\alpha$ (Algorithms \ref{pcskel} and \ref{orient}), and by
$G_{\mathrm{CPDAG};n}$ the true underlying CPDAG. Then, assuming (A)-(E)
and for $\alpha_n=2(1-\Phi(n^{1/2}c_n/2))$:
\begin{eqnarray}\label{consi0.pc}
P[\hat{G}_{\mathrm{CPDAG};n}(\alpha_n) = G_{\mathrm{CPDAG};n}] \to 1\ (n
\to \infty). 
\end{eqnarray}

Concerning the consistency of DAG based estimation of the concentration
matrix, we have the following new result.
\begin{lemma}
Under assumptions (A)-(D) and (F) the following holds. For any DAG $G$ in
the equivalence class of the true underlying CPDAG and using the estimator
$\hat{\Sigma}^{-1}_{G}$ in (\ref{estcov}): 
\label{lemma:consistenz}
\begin{eqnarray*}
  \label{eq:consistenz}
  \sup_{i,j}\left|\hat{\Sigma}^{-1}_{G,n;i,j}-\Sigma^{-1}_{n;i,j}\right|\xrightarrow{P}
  0\ (n\rightarrow \infty).  
\end{eqnarray*}
\end{lemma}
A proof is given in the Appendix. 
We then obtain the main theoretical result. 

\begin{thm}
Under assumptions (A)-(F) and using the tuning parameter
$\alpha_n=2(1-\Phi(n^{1/2}c_n/2))$ in the PC-algorithm, the 
following holds for the estimator in (\ref{DAGest-cov}): 
\label{thm:consistenz}
\begin{eqnarray*}
\sup_{i,j}\left|\hat{\Sigma}^{-1}_{\hat{G}_{\mathrm{DAG}}(\alpha),n;i,j} -
    \Sigma^{-1}_{n;i,j} \right|\xrightarrow{P} 0\ (n\rightarrow \infty). 
\end{eqnarray*}
\end{thm}

Proof: The estimate $\hat{G}_{\mathrm{DAG};n}(\alpha)$ is a DAG element of the
estimated equivalence class encoded by the estimated CPDAG 
$\hat{G}_{\mathrm{CPDAG};n}(\alpha)$. Denote this DAG by $G_*$. Consider
the event  
\begin{eqnarray*}
A_n = \{\hat{G}_{\mathrm{CPDAG};n}(\alpha_n) = G_{\mathrm{CPDAG};n}\},
\end{eqnarray*}
whose probability $\PR{A_n} \to 1\ (n \to \infty)$, see
(\ref{consi0.pc}). On $A_n$, $G_*$ must be a DAG element of the true
equivalence class $G_{\mathrm{CPDAG};n}$ and hence on $A_n$, Lemma
\ref{lemma:consistenz} yields consistency:
\begin{eqnarray*}
\sup_{i,j} \left|\hat{\Sigma}^{-1}_{\hat{G},n;i,j} - \Sigma^{-1}_{n;i,j} \right| =
\sup_{i,j} \left|\hat{\Sigma}^{-1}_{G_*,n;i,j} - \Sigma^{-1}_{n;i,j}
\right| \xrightarrow{P} 0\ (n\rightarrow \infty).
\end{eqnarray*}
Since $P[A_n] \to 1$, the proof is complete.\hfill$\Box$ 

\section{Simulation and real data analysis}\label{simu}

We examine the behavior of our PC-DAG estimator using simulated and
real data and compare it to the Glasso method \cite{FHT07,BGA08}.
The Glasso is defined as:

\begin{equation}
\label{logliki}
\hat{\Sigma}_{\mathrm{Glasso}}^{-1} = \argmin_{\Sigma^{-1}\
  \mbox{non-neg. def.}}
(-\log{\det{\Sigma^{-1}}}+\trace{(\hat{\Sigma}^{MLE}\Sigma^{-1})}
+\lambda\|\Sigma^{-1}\|_1)    
\end{equation}
where $\hat{\Sigma}^{MLE}$ is the empirical
covariance matrix in (\ref{MLEcov}), $\|\Sigma^{-1}\|_1 = \sum_{i<j}
|\Sigma_{ij}^{-1}|$ and the minimization is over non-negative definite 
matrices. 
   
All computations are done with the \texttt{R}-packages \texttt{pcalg}
\cite{pcalgMAN} and \texttt{glasso}. 

\subsection{Simulation study}

We consider a DAG and a non-DAG model for generating the data.  

\subsubsection{DAG models}\label{dagmodel}

We focus on the following class of DAG models. We generate recursively
\begin{equation*}
\begin{split}
  &X_1 = \epsilon_1\sim \mathcal{N}(0,1),\\
  &X_i=\sum_{r=1}^{i-1}B_{ir}X_r+\epsilon_i\ (i=2,\dots,p),
\end{split}
\end{equation*}
where $\epsilon_1,\ldots,\epsilon_p$ i.i.d. $\sim \mathcal{N}(0,1)$ and
$B$ is an adjacency matrix generated as follows. We first fill
the matrix $B$ with zeros and replace every matrix entry in the
lower triangle by independent realizations of Bernoulli($s$) random
variables with success probability $s$ where $0<s<1$. 
Afterwards, we replace each entry having a 1 in the matrix $B$ by
independent realizations of a Uniform([0.1,1]) random variable. 
If $i<j$ and $B_{ji}\neq 0$ the corresponding DAG has a directed edge from
node $i$ to node $j$. The variables $X_1,\ldots ,X_p$ have a multivariate
Gaussian distribution with mean zero and covariance $\Sigma$ which can be
computed from $B$. We consider this model for different settings of $n$,
$s$ and $p$:

\begin{itemize}
\item[D1:] $n=30$, $s=0.01$, $p=40,50,60,70,80,90,100,110,120$ 
\item[D2:] $n=50$, $s=0.01$, $p=40,50,60,70,80,90,100,110,120$ 
\item[D3:] $n=30$, $s=0.05$, $p=40,50,60,70,80,90,100,110,120$
\item[D4:] $n=50$, $s=0.05$, $p=40,50,60,70,80,90,100,110,120$ 
\end{itemize}

The settings D1 to D4 mainly differ in the sparsity $s$ of the generated
data, which is related to the expected neighborhood size $\ERW{adj(i,G)}=s(p-1)$
for all $i$. For each of these settings we estimate the covariance and the
concentration matrix with both methods, our PC-DAG and the Glasso estimator.
 
We use two different performance measures
to compare the two estimation techniques. First, the Frobenius norm of the
difference between the estimated and the true matrix $\|\hat{\Sigma}-\Sigma\|_F$
and $\|\hat{\Sigma}^{-1}-\Sigma^{-1}\|_F$.
And second, the Kullback-Leibler Loss $\Delta_{KL}(\hat{\Sigma}^{-1},\Sigma^{-1})
=tr(\Sigma\hat{\Sigma}^{-1})-\log|\Sigma\hat{\Sigma}^{-1}|-p$.  

We sample data $X^{(1)},\ldots, X^{(n)}$ i.i.d. from the DAG model
described above for each value 
of $p$ in settings D1-D4. Then we derive, on a separate validation data-set
${X^{(1)}}^*,\ldots,{X^{(n)}}^*$, the optimal value of the tuning parameters
$\alpha$ (PC-DAG) or $\lambda$ (Glasso), with respect to the
negative Gaussian log-likelihood. The two different 
performance measures are evaluated for the estimates based on the training
data $X^{(1)},\ldots, X^{(n)}$ with optimal tuning parameter choice
based on the validation data. All 
results are based on 50 independent simulation runs.

\begin{figure}[h]
        \centerline{ 
          \subfigure[For setting D1]{\includegraphics[scale=0.3]{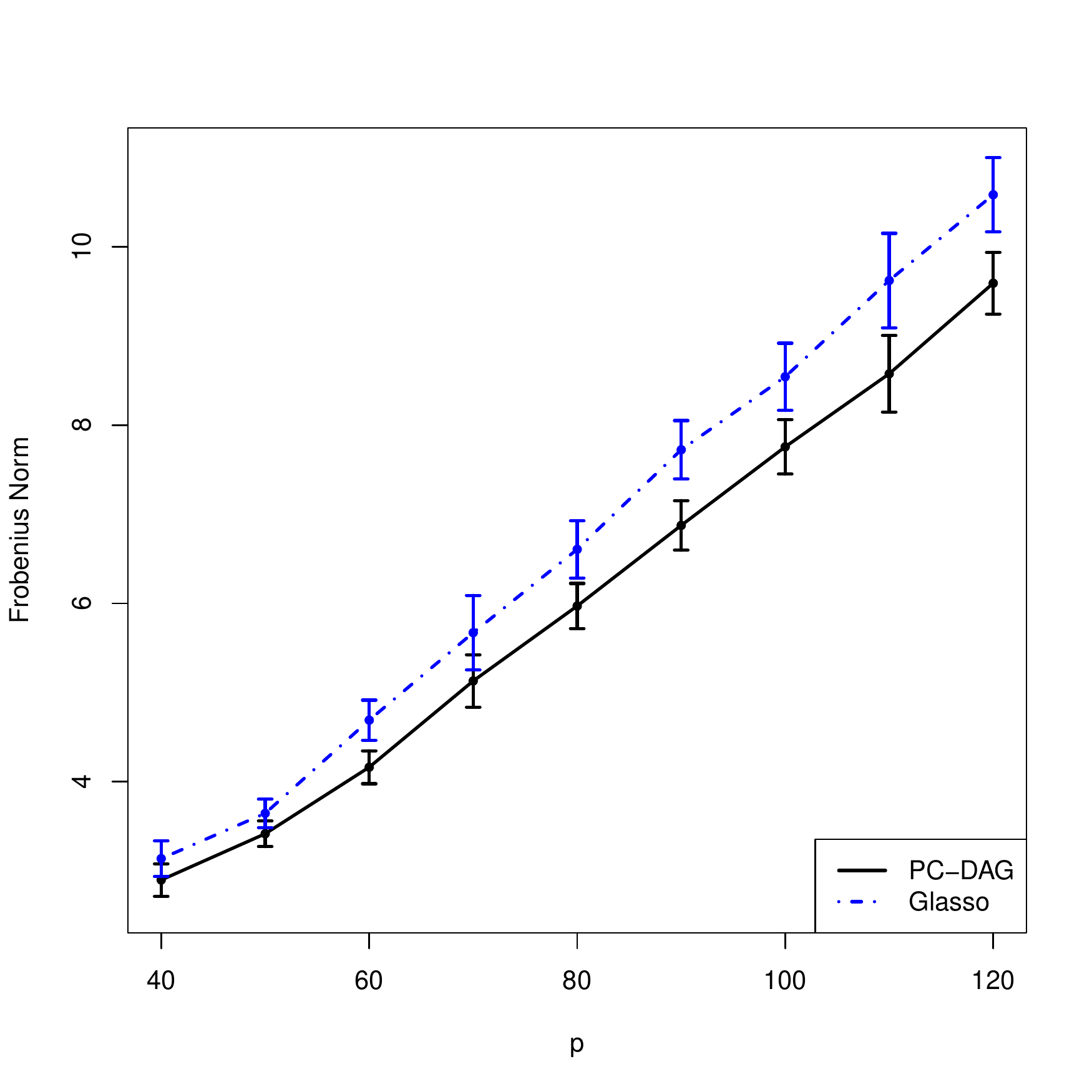}} 
          \subfigure[For setting D2]{\includegraphics[scale=0.3]{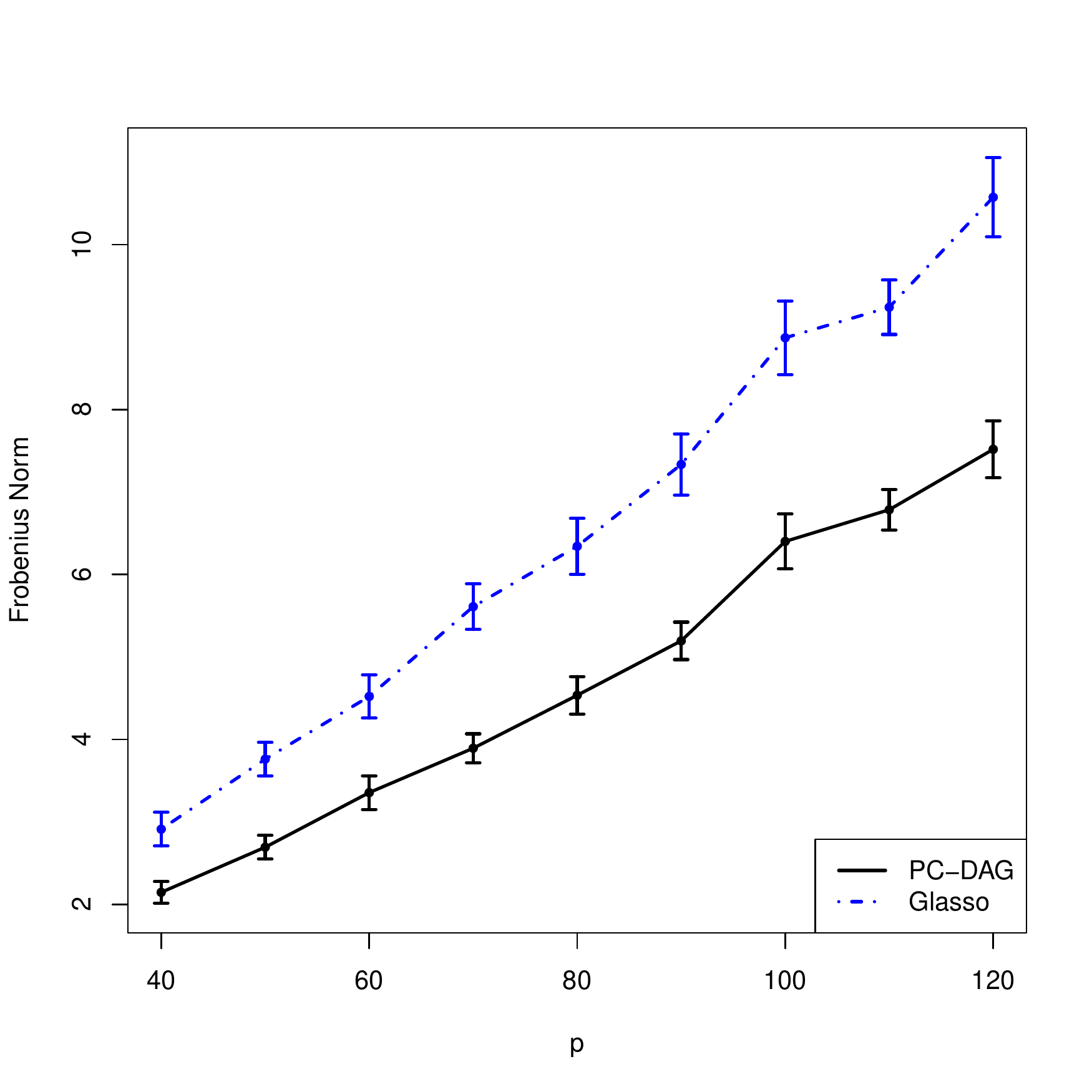}}} 
        \centerline{ 
          \subfigure[For setting D3]{\includegraphics[scale=0.3]{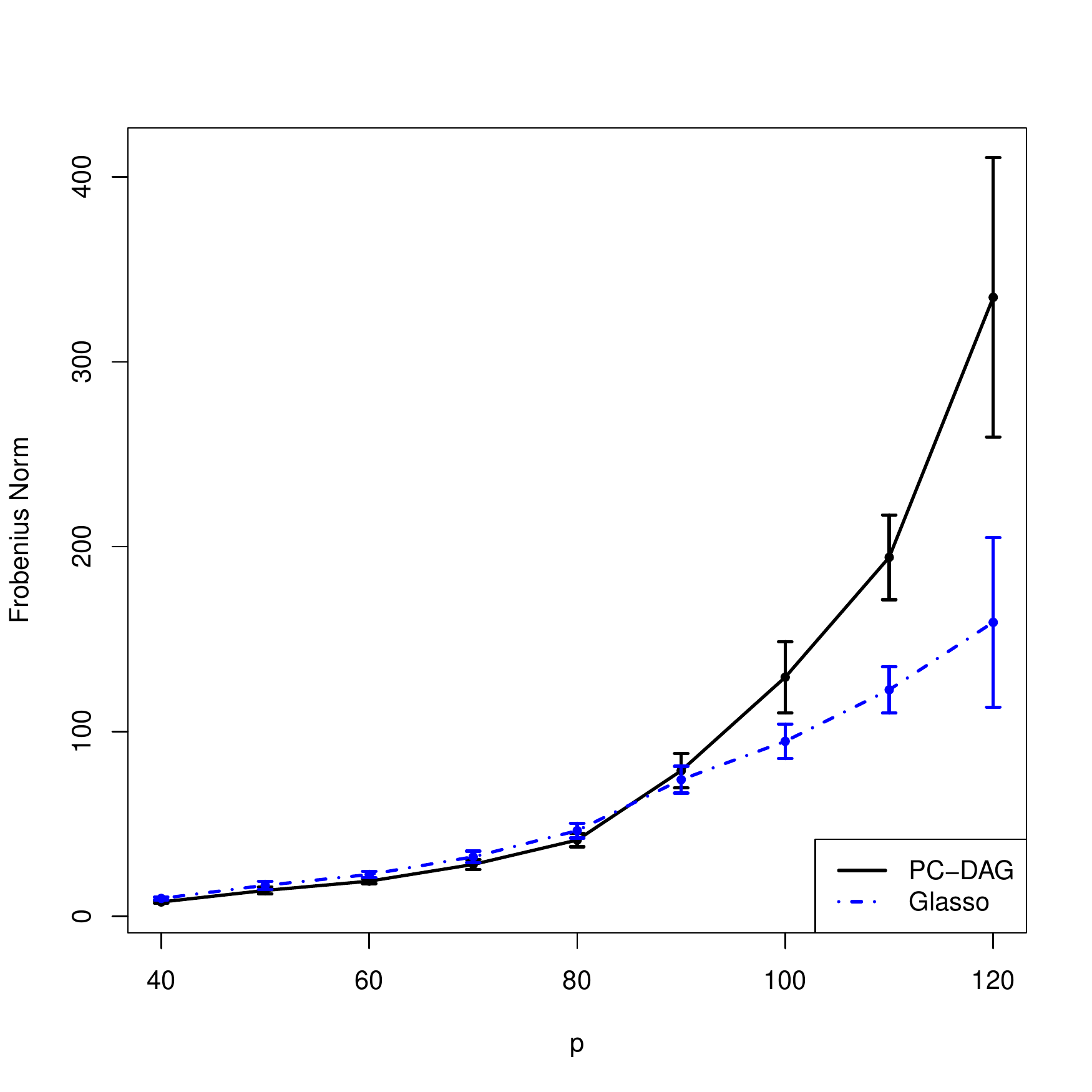}}
          \subfigure[For setting D4]{\includegraphics[scale=0.3]{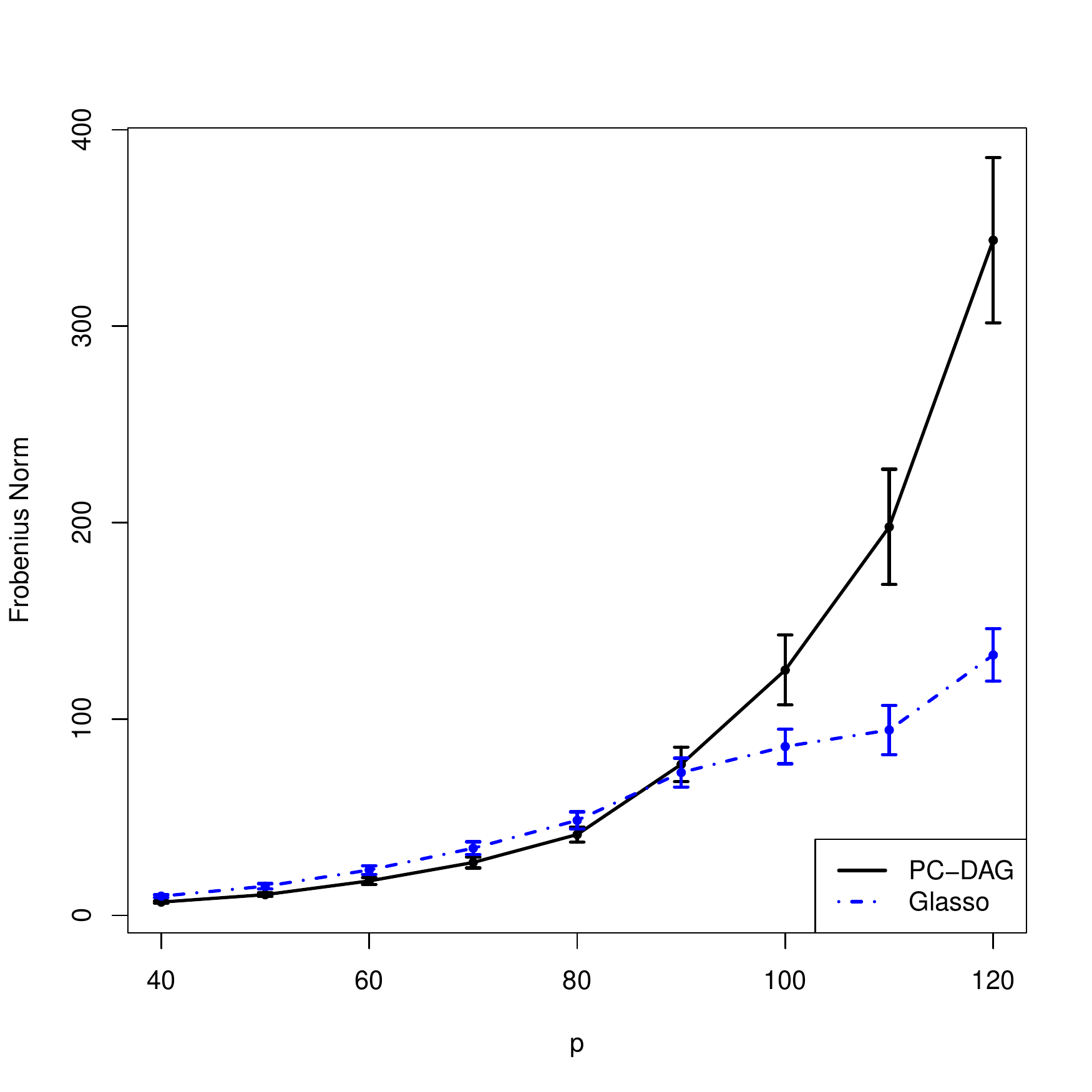}}}
        \caption{Plots of $\|\hat{\Sigma}-\Sigma\|_F$
           for DAG models. Vertical bars indicate (pointwise) 95\%
           confidence intervals. 
        }
        \label{D1D2D3D4}
\end{figure}

%%%%%%%%%%%%%%%%%%%%%%%%%%%%%%%%%%%%%%%%%%%%%%%%%%%%%%%%%%%%%%%%%%%%%%%%%%%%%%%%%
\begin{figure}[h]
        \centerline{ 
          \subfigure[For setting D1]{\includegraphics[scale=0.3]{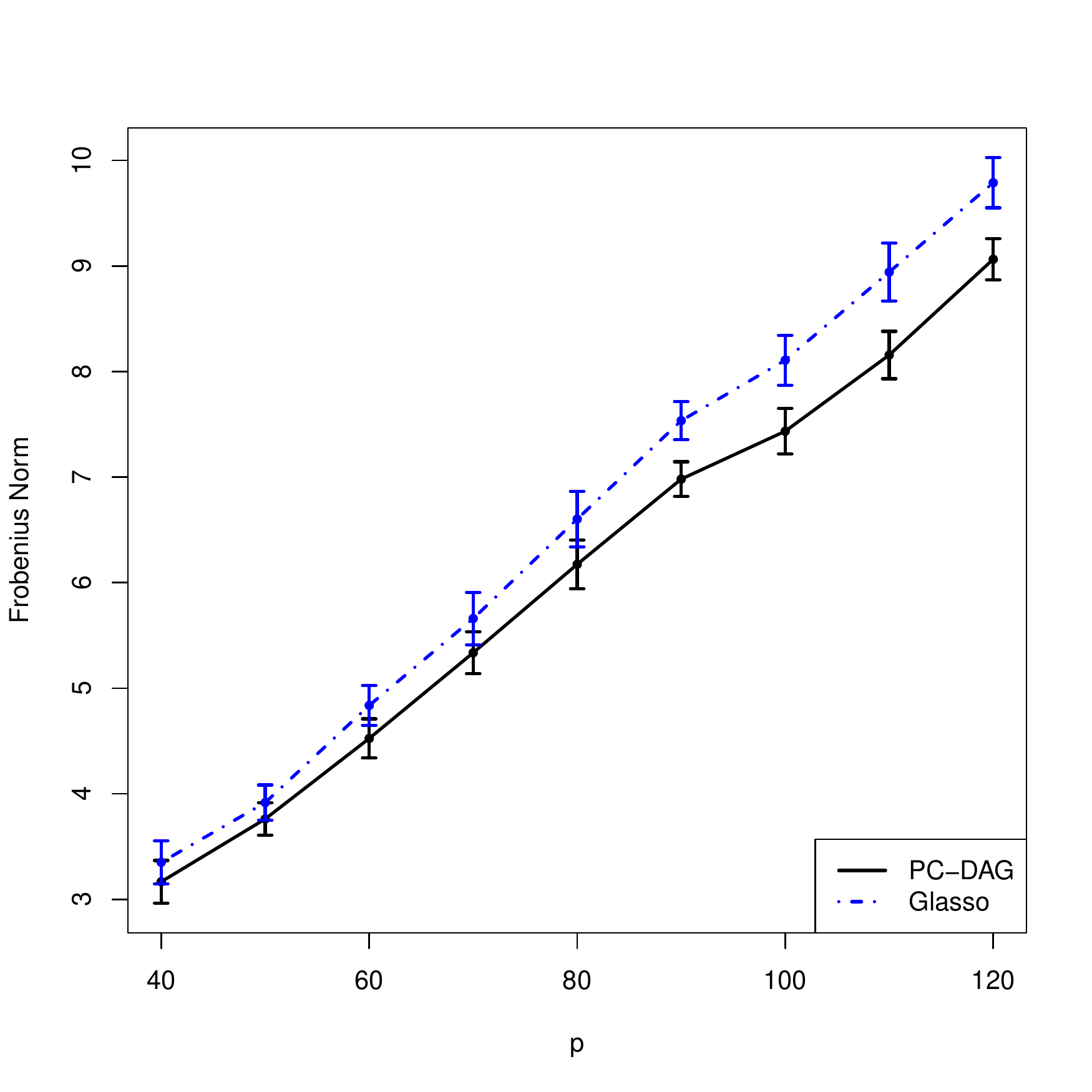}}
          \subfigure[For setting D2]{\includegraphics[scale=0.3]{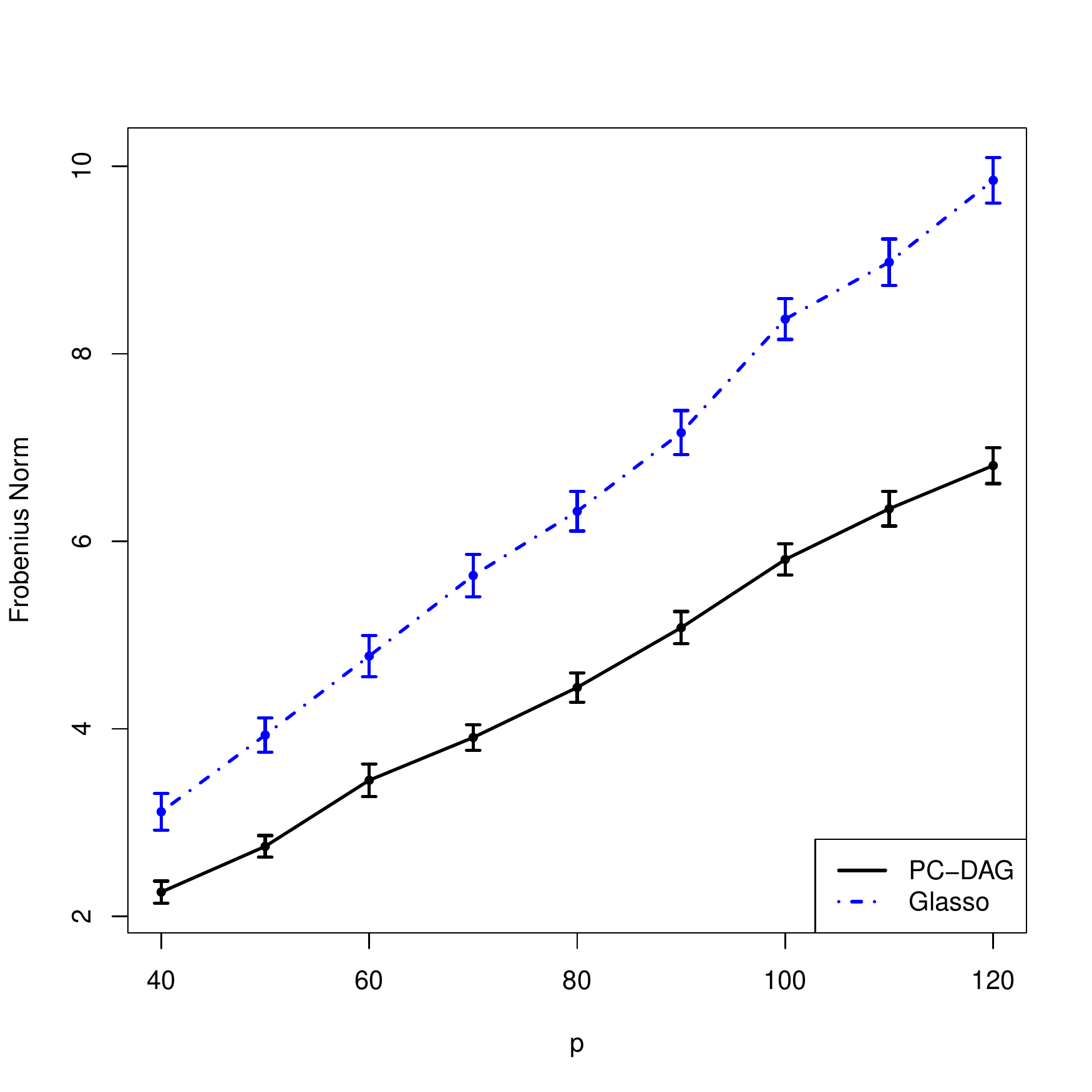}}}
        \centerline{ 
          \subfigure[For setting D3]{\includegraphics[scale=0.3]{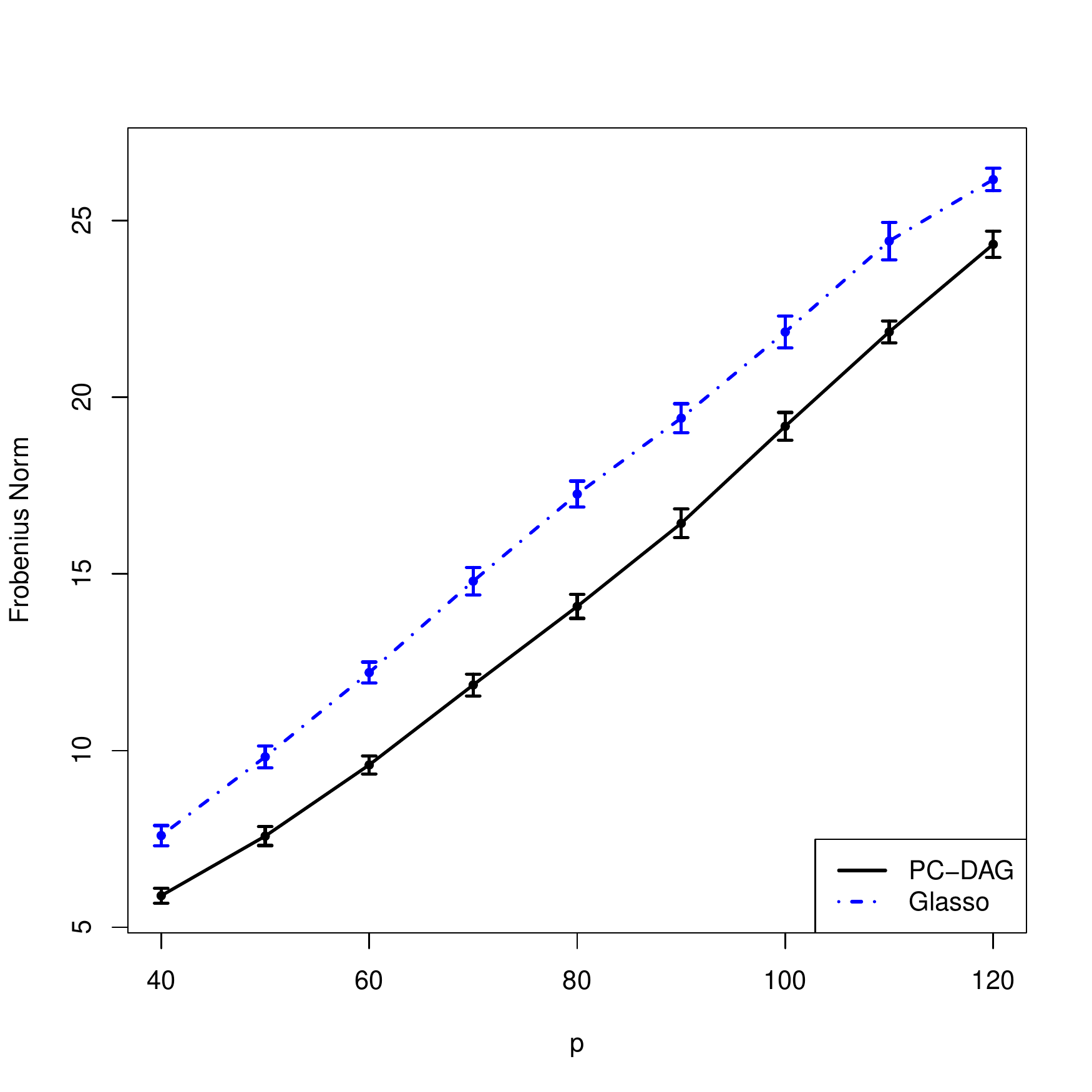}}
          \subfigure[For setting D4]{\includegraphics[scale=0.3]{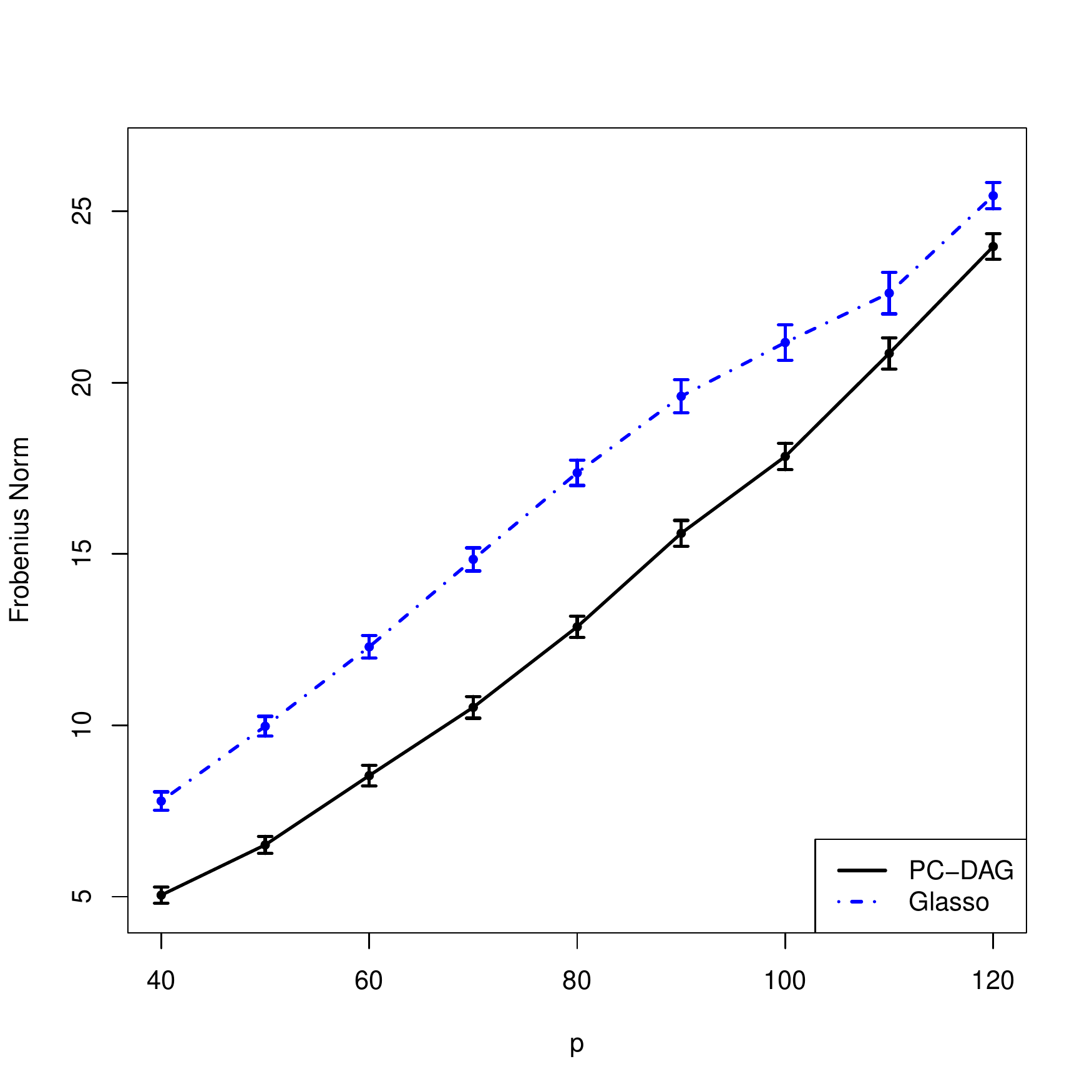}}}
        \caption{Plots of 
          $\|\hat{\Sigma}^{-1}-\Sigma^{-1}\|_F$ for DAG models. Vertical
          bars indicate (pointwise) 95\% confidence intervals.
        }
        \label{ID1D2D3D4}
\end{figure} 

Figures \ref{D1D2D3D4} and \ref{ID1D2D3D4} show that in the sparse settings
D1 and D2, the PC-DAG estimator
clearly outperforms Glasso. Concerning the more dense settings D3 and D4, the
PC-DAG method degrades only for the covariance matrix, whereas for the
inverse covariance matrix $\Sigma^{-1}$, the figures still show an improvement of
the PC-DAG estimator compared to the Glasso. If we match Figure
\ref{D1D2D3D4} (a) with Figure \ref{D1D2D3D4} (b) and Figure \ref{ID1D2D3D4} (a)
with Figure \ref{ID1D2D3D4} (b), we see that for a small increase of the
sample size the Glasso improves substantially less compared to the PC-DAG
estimator. The results in terms of the Kullback-Leibler loss are summarized
in Table \ref{tab1}.       

\subsubsection{Non DAG models}

Next we generate data from a non-DAG model proposed by
\cite{RBLZ08}. The
concentration matrix equals 
\[
\Sigma^{-1}=B+\delta I,
\]
where each off-diagonal entry in $B$ is generated independently and
equals 0.5 with probability $\pi$ or 0 with probability $1-\pi$, all
diagonal entries of $B$ are zero, and $\delta$ is chosen such that the condition
number of $\Sigma^{-1}$ is $p$. The concentration matrices, which we
generate from this model vary in their level of sparsity: for
$\Sigma^{-1}_{(1)}$ we take $\pi=0.1$ and for $\Sigma^{-1}_{(2)}$ we
choose $\pi=0.5$, i.e. $\Sigma^{-1}_{(1)}$ is sparser than
$\Sigma^{-1}_{(2)}$. Note that the expected numbers of non-zero entries in
$\Sigma^{-1}_{(1)}$ and $\Sigma^{-1}_{(2)}$ are proportional to $p^2$.

We generate Gaussian data $X^{(1)},\ldots,X^{(n)}$ i.i.d. $\sim$
$\mathcal{N}_p(0,\Sigma)$ with $\Sigma^{-1}$ constructed as above,
according to the following settings:  
\begin{itemize}
\item[nD1:] $n=30$, $\pi=0.1$, $p=40,50,60,70,80,90,100,110,120$ 
\item[nD2:] $n=50$, $\pi=0.1$, $p=40,50,60,70,80,90,100,110,120$
\item[nD3:] $n=30$, $\pi=0.5$, $p=40,50,60,70,80,90,100,110,120$ 
\item[nD4:] $n=50$, $\pi=0.5$, $p=40,50,60,70,80,90,100,110,120$ 
\end{itemize}

We tune and compare the estimation methods as described in Section
\ref{dagmodel}. 

\begin{figure}[h]
        \centerline{ 
          \subfigure[For setting nD1]{\includegraphics[scale=0.3]{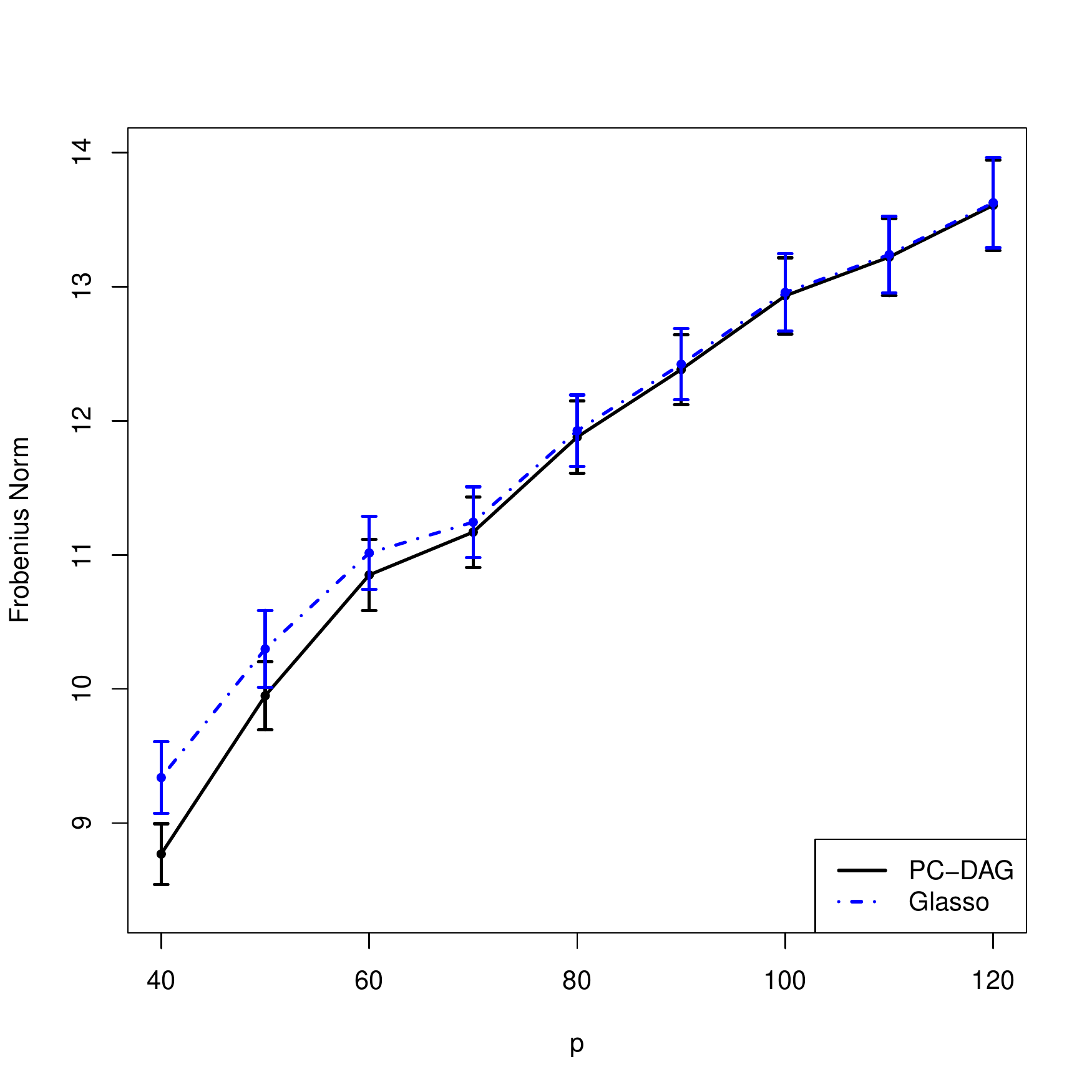}}
          \subfigure[For setting nD2]{\includegraphics[scale=0.3]{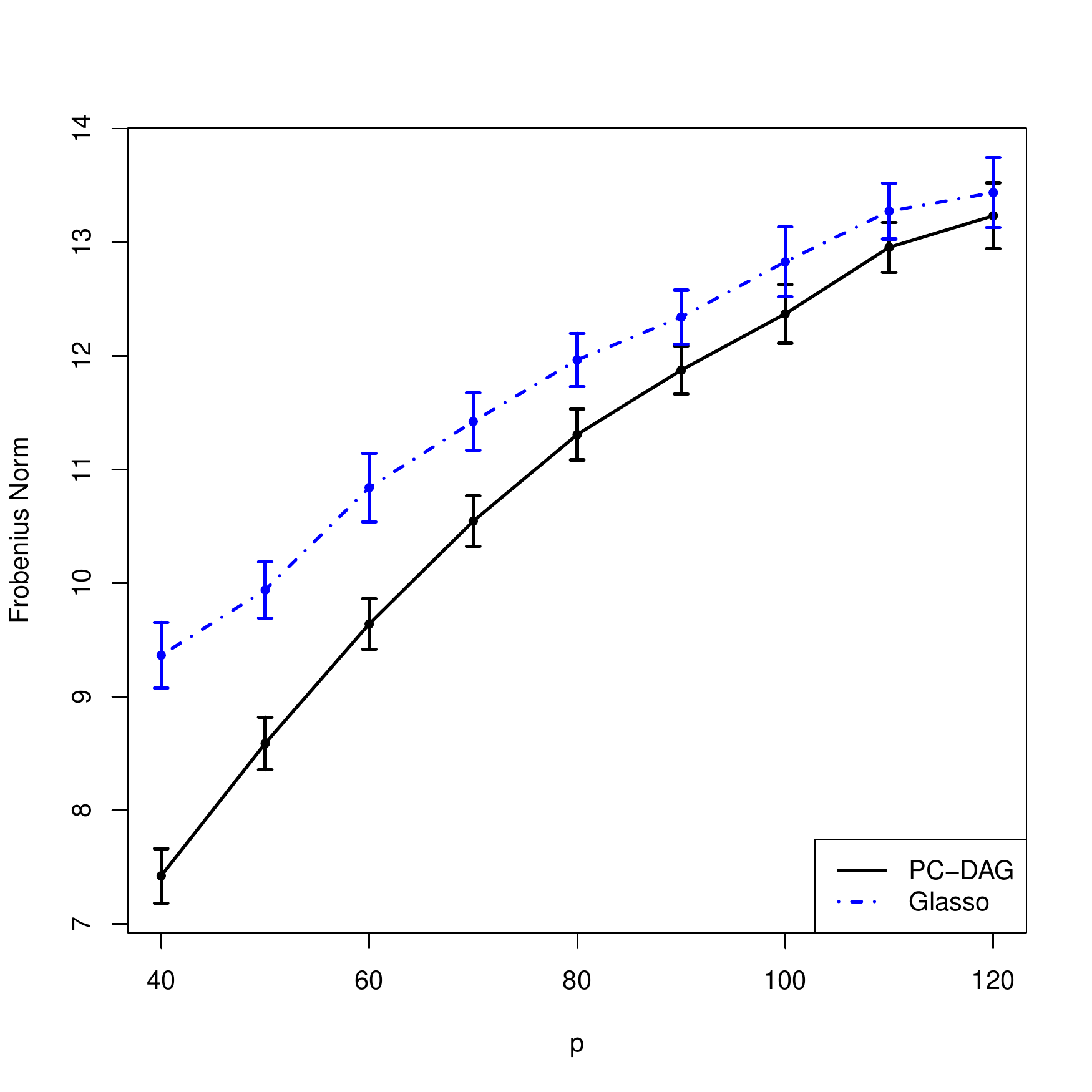}}}
        \centerline{ 
          \subfigure[For setting nD3]{\includegraphics[scale=0.3]{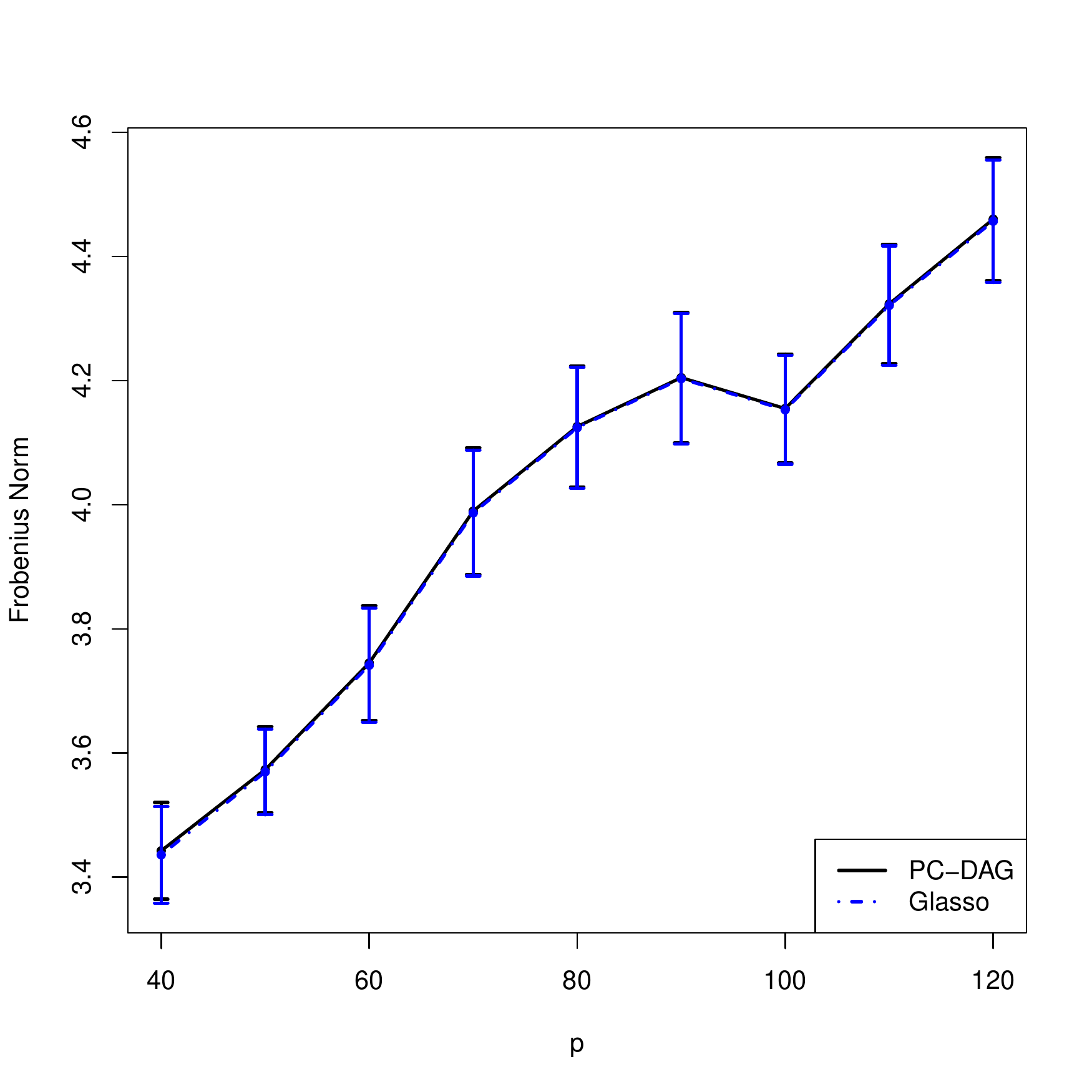}}
          \subfigure[For setting nD4]{\includegraphics[scale=0.3]{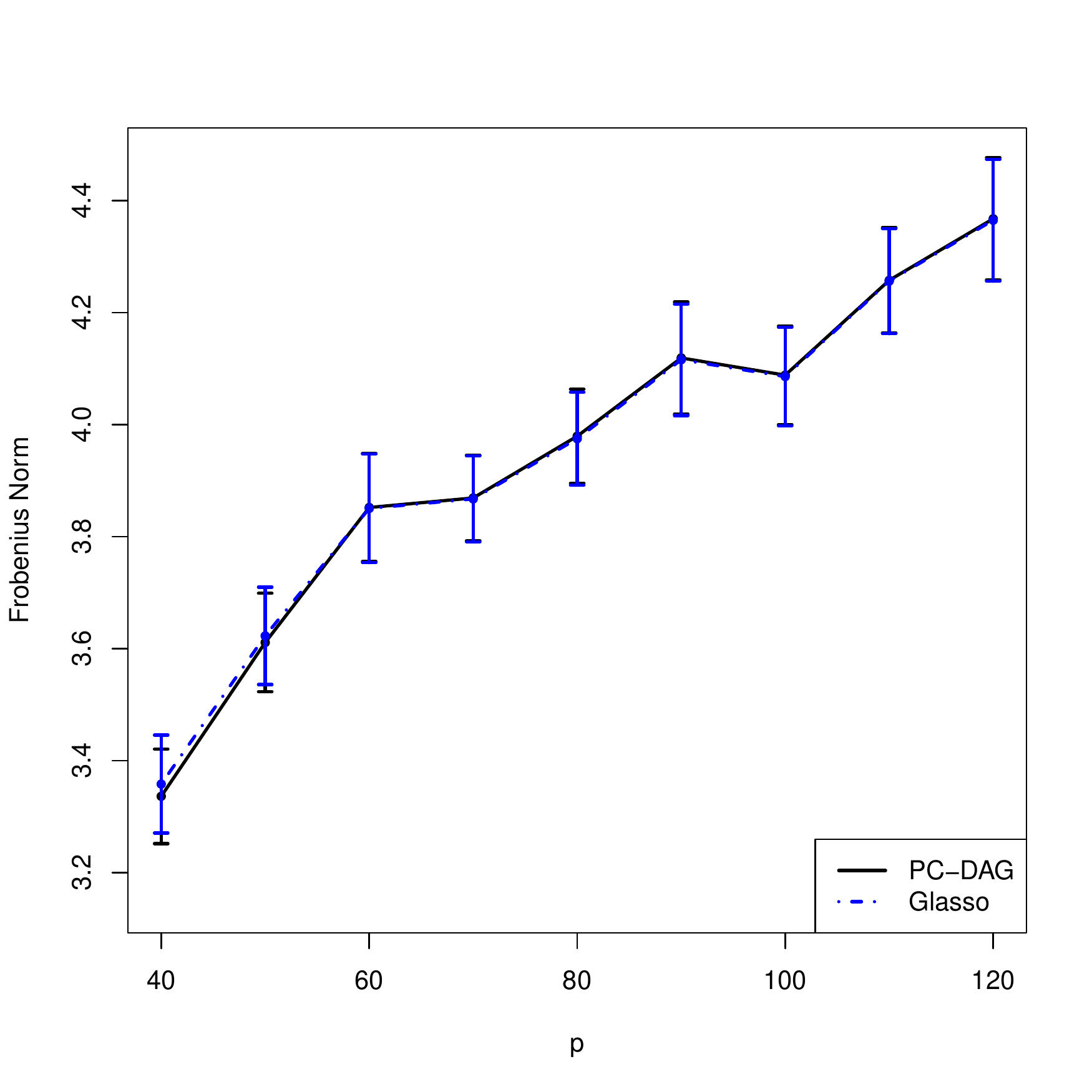}}}
        \caption{Plots of 
          $\|\hat{\Sigma}-\Sigma\|_F$ for non DAG models. Vertical
          bars indicate (pointwise) 95\% confidence intervals. 
        }
        \label{nD1nD2}
\end{figure} 

\begin{figure}[h]
        \centerline{ 
          \subfigure[For setting nD1]{\includegraphics[scale=0.3]{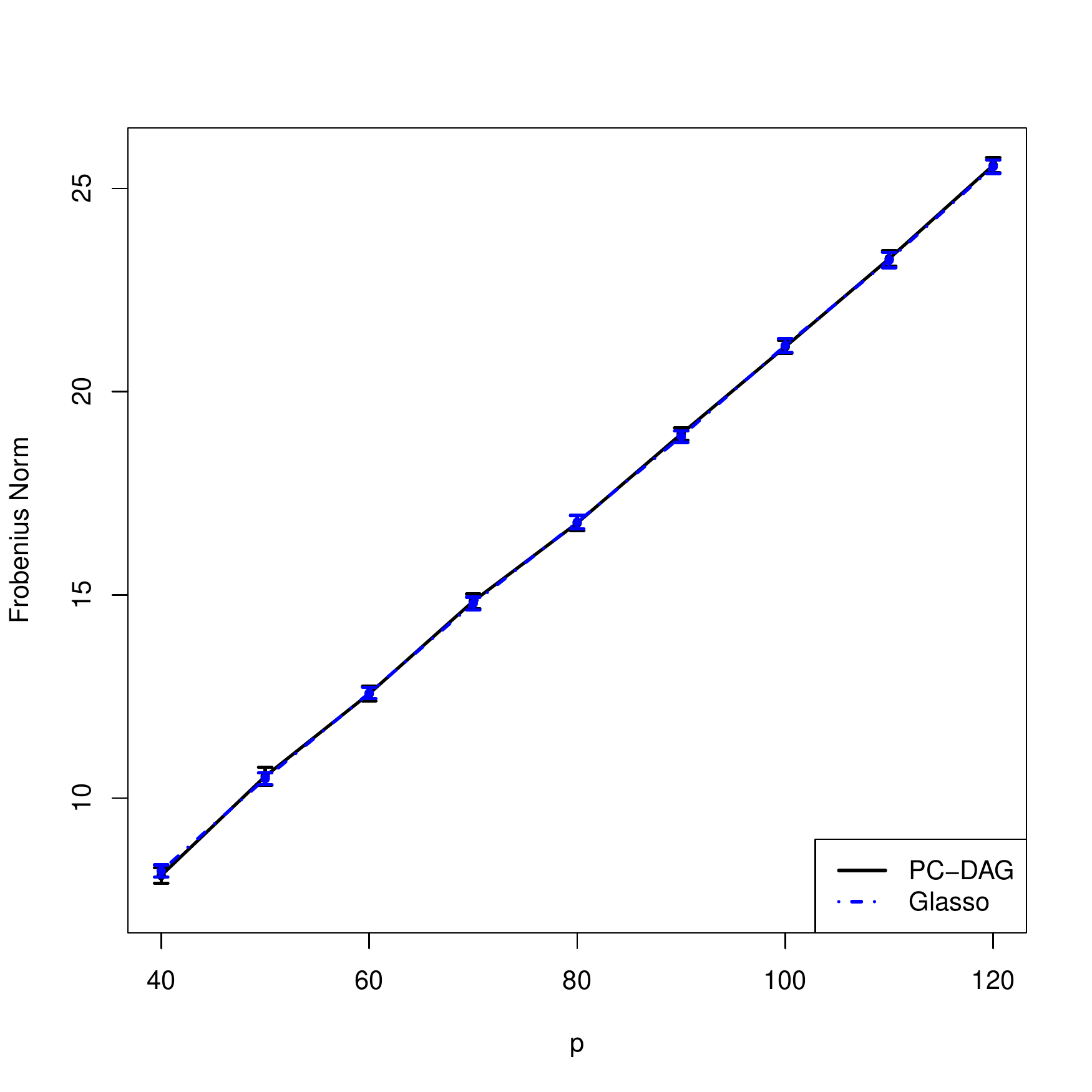}}
          \subfigure[For setting nD2]{\includegraphics[scale=0.3]{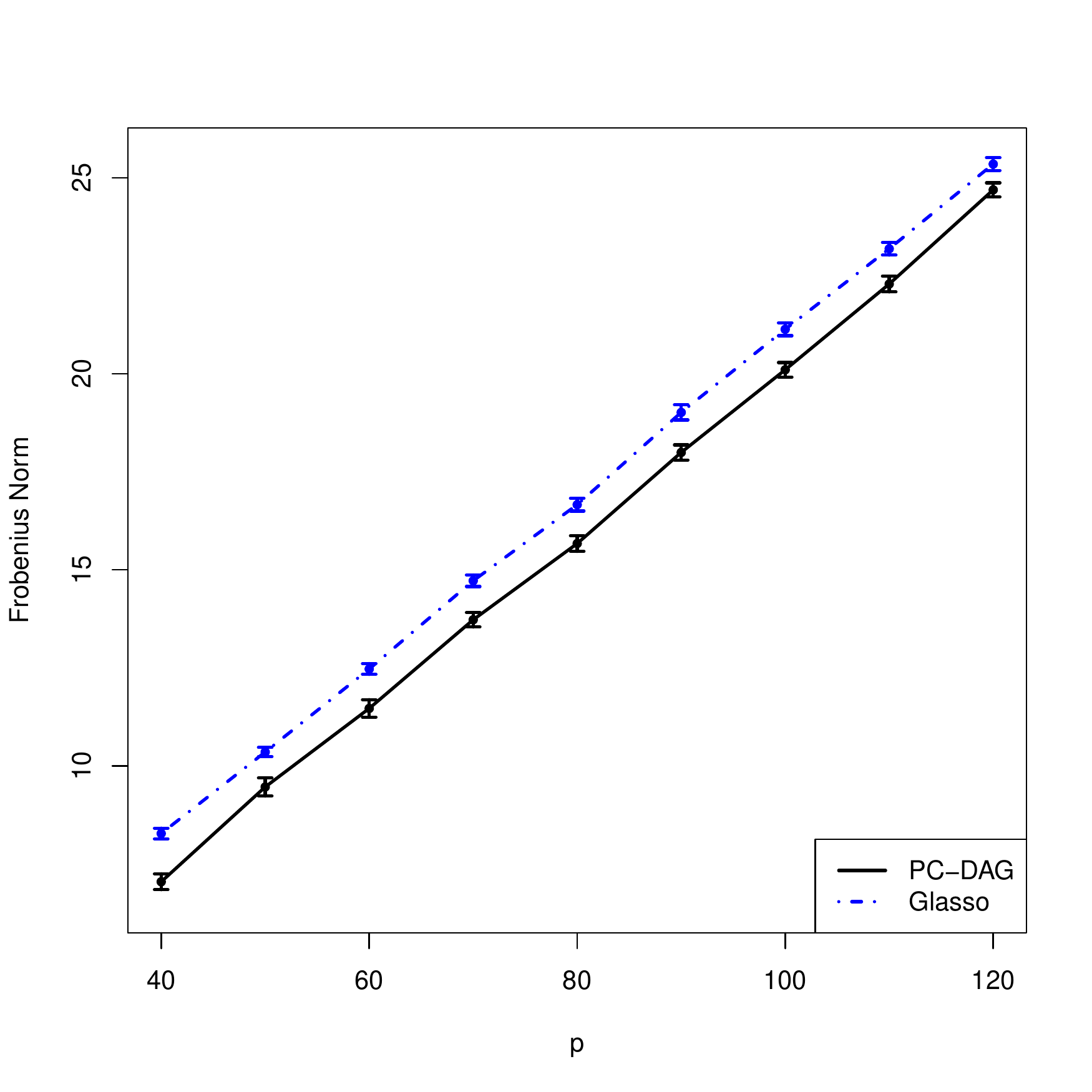}}}
        \centerline{ 
          \subfigure[For setting nD3]{\includegraphics[scale=0.3]{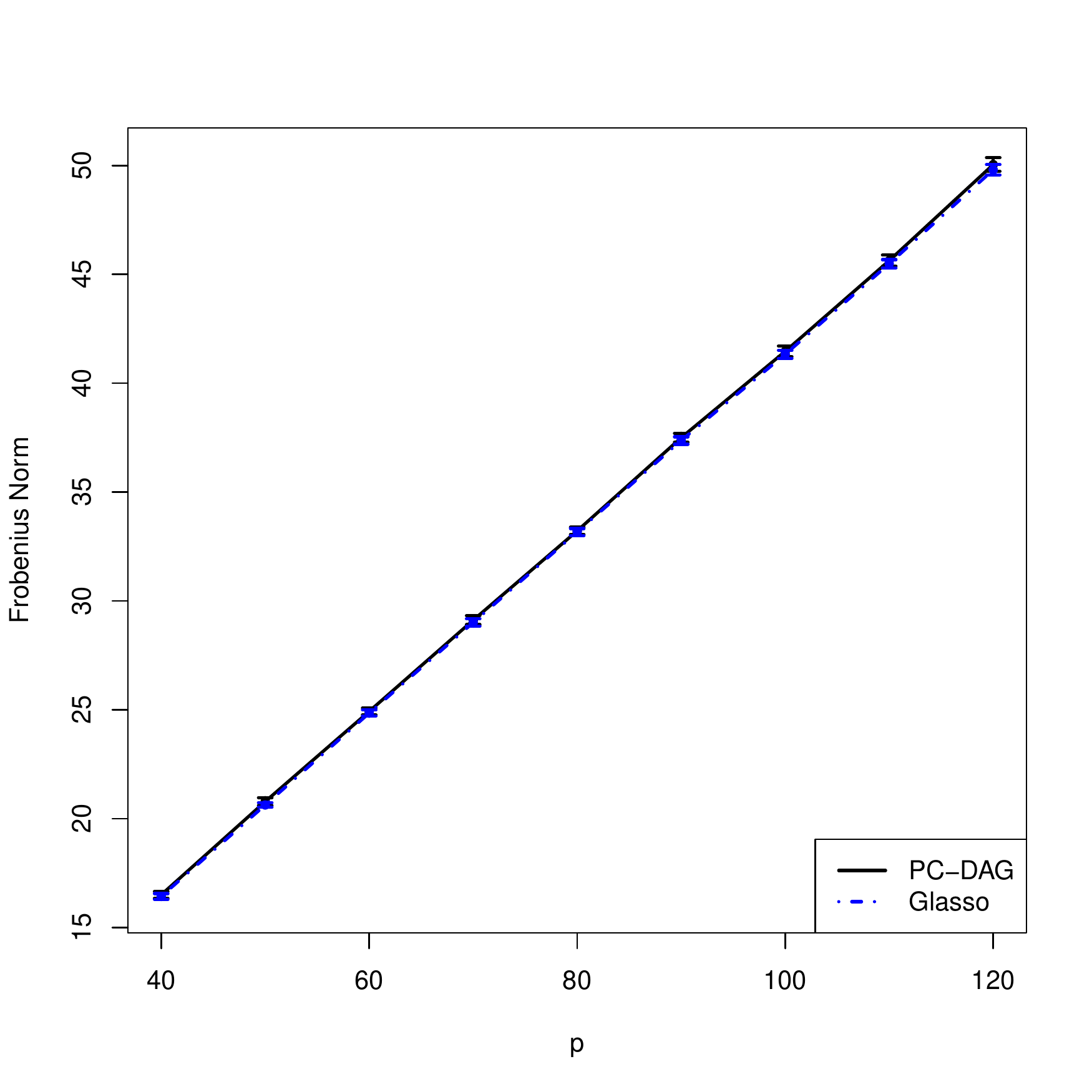}}
          \subfigure[For setting nD4]{\includegraphics[scale=0.3]{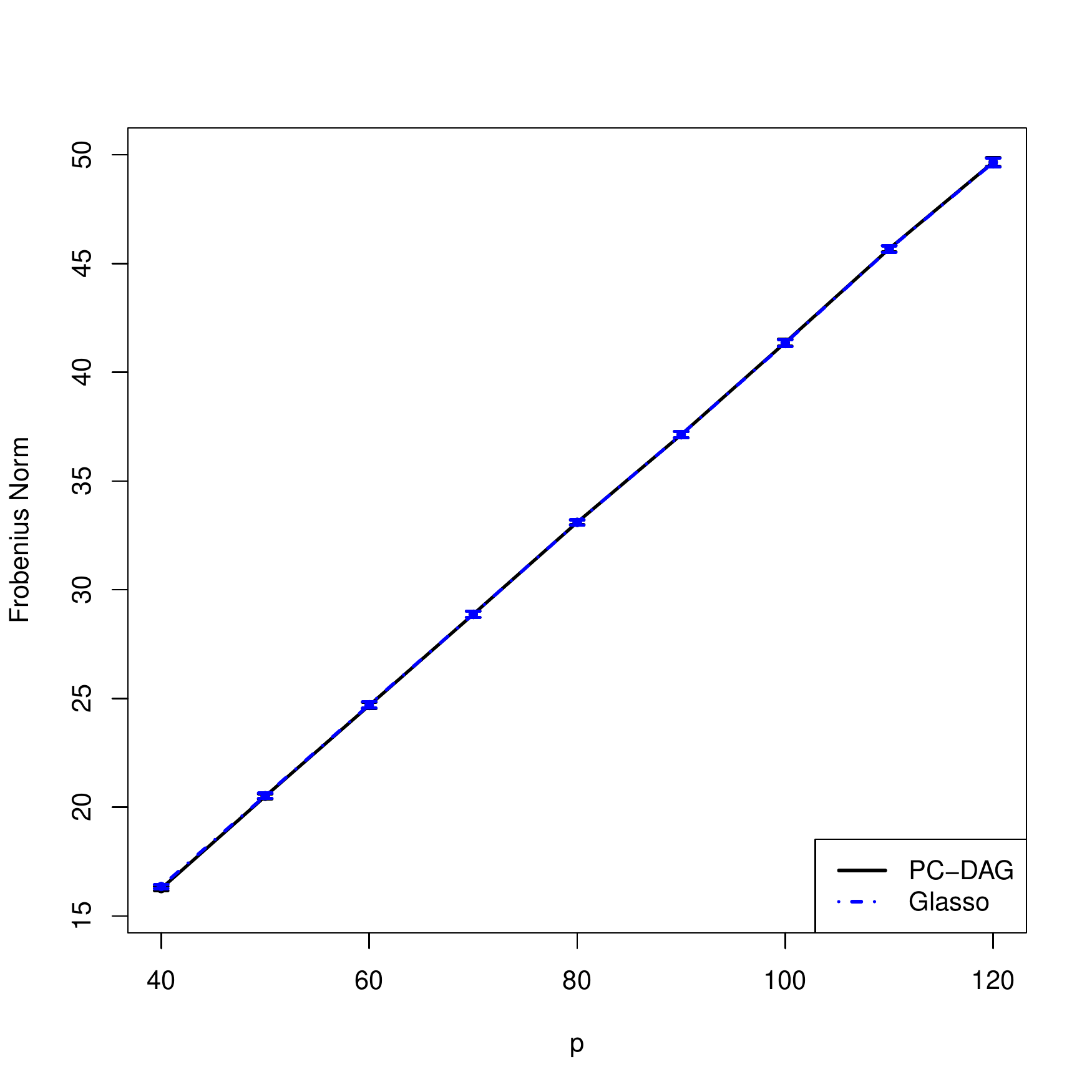}}}
        \caption{Plots of 
          $\|\hat{\Sigma}^{-1}-\Sigma^{-1}\|_F$ for  non DAG models. Vertical
          bars indicate (pointwise) 95\% confidence intervals. 
        }
        \label{InD1nD2}
\end{figure} 

In Figures \ref{nD1nD2} and \ref{InD1nD2} we see that in case of
the dense model with $\pi = 0.5$, the two methods
do not differ much (some of the differences are so small that they are
invisible on the scales shown in the plots). But for the sparse model with
$\pi = 0.1$ we observe 
that our PC-DAG estimator is better than the Glasso, in particular for the
setting nD2. The results in terms of the Kullback-Leibler loss are summarized in
Table \ref{tab1}.   

 \begin{table}
 \centering
 \begin{tabular}[h]{|c|c|c|c|c|}\hline
 \multicolumn{5}{|c|}{Kullback-Leibler Loss}\\\hline
 \multicolumn{5}{|c|}{DAG}\\\hline
 \multicolumn{5}{|c|}{$n=30$}\\\hline
 $p$ &\multicolumn{2}{|c|}{$s=0.01$ (D1)} & \multicolumn{2}{|c|}{$s=0.05$ (D3)} \\\hline 
      & Glasso & PC-DAG& Glasso& PC-DAG \\\hline
  40 & 3.78(0.17)  & 3.38(0.16) & 13.64(0.41) & 9.27(0.29)  \\\hline 
  80 & 12.75(0.34) & 11.36(0.29) & 54.63(0.9) & 41.69(0.67) \\\hline 
  120 & 25.5(0.41)  & 22.93(0.42) & 79.34(1.35) & 104.43(1.47) \\\hline 
 \multicolumn{5}{|c|}{DAG}\\\hline
 \multicolumn{5}{|c|}{$n=50$}\\\hline
 $p$ &\multicolumn{2}{|c|}{$s=0.01$ (D2)} & \multicolumn{2}{|c|}{$s=0.05$ (D4)} \\\hline 
      & Glasso & PC-DAG& Glasso& PC-DAG \\\hline
  40 & 3.12(0.15)  & 1.88(0.08) & 13.3(0.31) & 6.26(0.18)   \\\hline 
  80 & 11.07(0.26) & 6.32(0.17) & 53.08(1.22) & 31.83(0.53)  \\\hline 
  120 & 24.35(0.47)  & 13.76(0.27) & 66.21(2.78) & 87.11(0.93)  \\\hline 
\multicolumn{5}{|c|}{non DAG}\\\hline
 \multicolumn{5}{|c|}{$n=30$}\\\hline
 $p$ &\multicolumn{2}{|c|}{Model $\Sigma^{-1}_{(1)}$ (nD1)} &
 \multicolumn{2}{|c|}{Model $\Sigma^{-1}_{(2)}$ (nD3)} \\\hline 
      & Glasso & PC-DAG& Glasso& PC-DAG \\\hline
  40 & 15.61(0.21)  & 14.91(0.22) &13.53(0.16) &13.71(0.16)   \\\hline 
  80 & 35.63(0.45) & 35.9(0.49) &29.36(0.33) &29.49(0.33)  \\\hline 
  120 & 56.44(0.67)  & 56.88(0.7) &45.34(0.45) &45.76(0.46)  \\\hline 
\multicolumn{5}{|c|}{non DAG}\\\hline
 \multicolumn{5}{|c|}{$n=50$}\\\hline
 $p$ &\multicolumn{2}{|c|}{Model $\Sigma^{-1}_{(1)}$ (nD2)} &
 \multicolumn{2}{|c|}{Model $\Sigma^{-1}_{(2)}$ (nD4)} \\\hline 
      & Glasso & PC-DAG& Glasso& PC-DAG \\\hline
  40 & 15.38(0.24)  & 10.58(0.18) &12.76(0.16) &12.91(0.17)   \\\hline 
  80 & 34.28(0.4) & 32.13(0.3) &27.49(0.28) &27.68(0.34)  \\\hline 
  120 & 53.69(0.7)  & 53.16(0.67) &42.85(0.5) &43.08(0.5)  \\\hline 
 \end{tabular}
 \caption{Kullback-Leibler Loss (standard error in parentheses).}
 \label{tab1}
 \end{table}      

\subsection{Real data}

In this section we compare the two estimation methods for real data. 

\subsubsection{Isoprenoid gene pathways in \emph{Arabidopsis thaliana}}

We analyze the gene expression data from the isoprenoid
biosynthesis pathway in \emph{Arabidopsis thaliana} given in \cite{AWPB04}.
Isoprenoids comprehend the most diverse class of natural products and have
been identified in many different organisms. In plants isoprenoids play
important roles in a variety of processes such as photosynthesis,
respiration, regulation of growth and development.\\ 
This data set consists of $p=39$ isoprenoid genes for which we have $n=118$
gene expression patterns under various experimental conditions. As
performance measure we use the 10-fold cross-validated negative Gaussian
log-likelihood for centered data.\\

\begin{figure}[h]
        \centerline{         
            \includegraphics[scale=0.45]{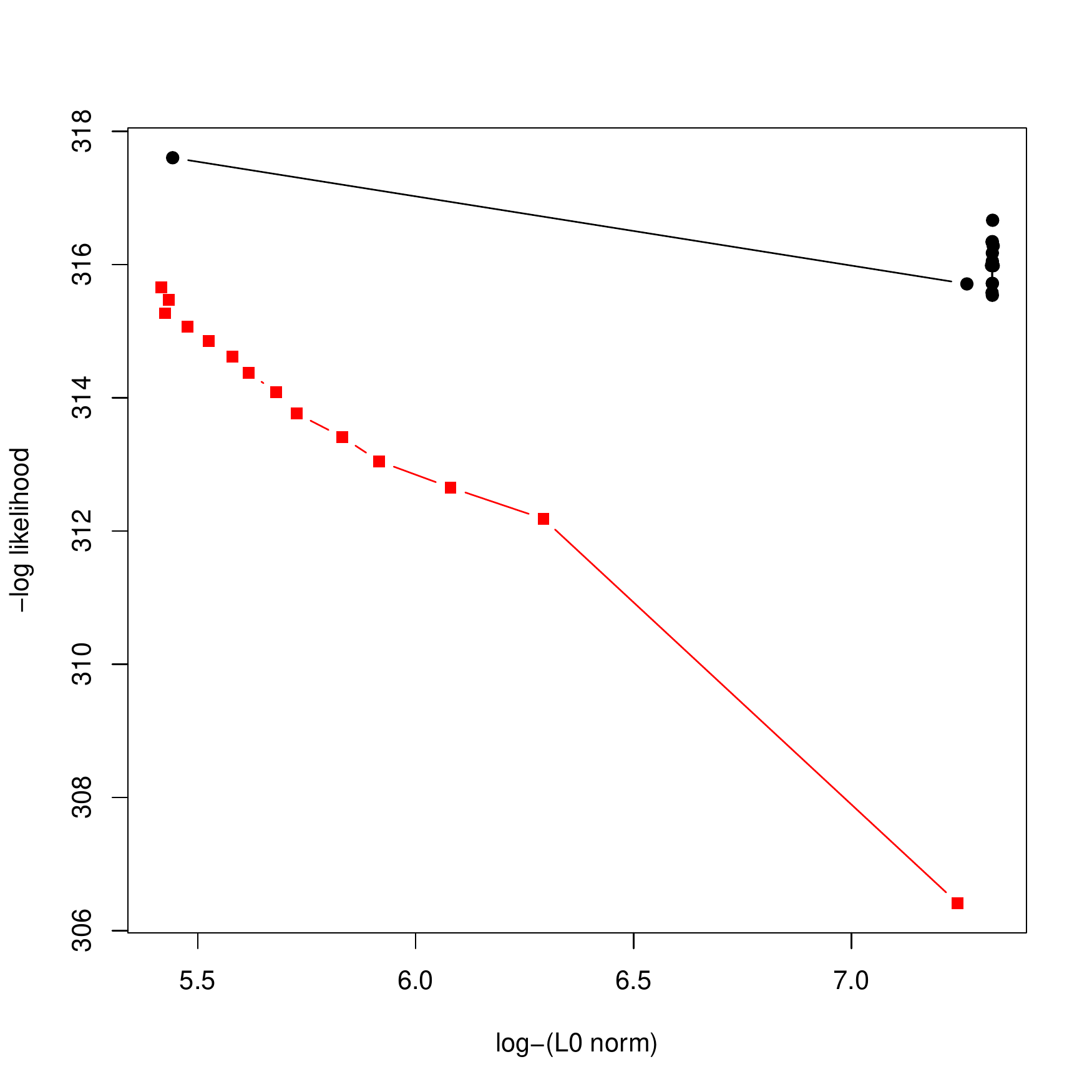}}
        \caption{10-fold CV of negative log-likelihood against the
          logarithm of the average number of non-zero entries of the
          estimated concentration matrix 
$\hat{\Sigma}^{-1}$. The squares stand for the Glasso and the circles for
the PC-DAG estimator.} 
        \label{genom}
\end{figure}

The results are described in Figure \ref{genom}. 
We find that none of the two methods performs substantially
better than the other and the slight superiority of Glasso is in the order
of 1\% only. The marginal difference in the negative log-likelihood
between the two estimation techniques may be due to the high noise in the data. 

\subsubsection{Breast Cancer data}

Next, we explore the performance   
on a gene expression data set from breast tumor samples. The tumor samples were
selected 
from the Duke Breast Cancer SPORE tissue bank on the basis of several criteria.
For more details on the data set see \cite{WEST01}.
The data matrix monitors $p=7129$ genes in $n=49$ breast tumor
samples. We only use the 100 variables
having the largest sample variance.\\ 
As before we first center the data and then compute the negative
log-likelihood via 10-fold cross-validation. Figure
\ref{bdat} shows the result.

\begin{figure}[h]
         \centerline{ 
            \includegraphics[scale=0.45]{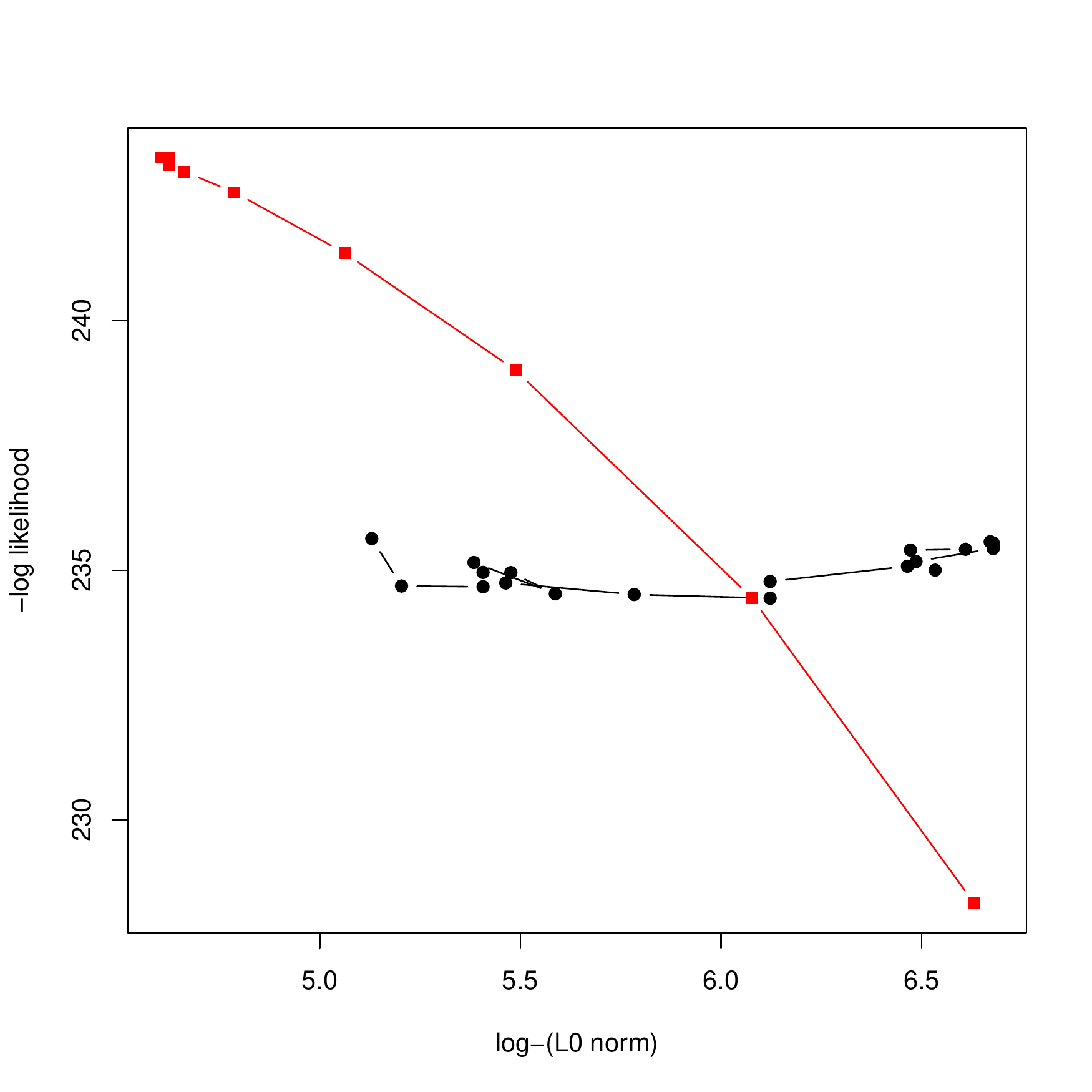}}
        \caption{
          10-fold CV of negative log-likelihood  against the
          logarithm of the average number of non-zero entries of the
          estimated concentration matrix. The squares stand for the Glasso
          and the circles for the PC-DAG estimator.} 
        \label{bdat}
\end{figure}
As for the Isoprenoid gene pathways data-set, we cannot nominate a winner
here. In fact, the performances are even more indistinct than before.

\section{A robust PC-DAG covariance estimator}
\label{robu}

In this section we propose a robust version of the PC-DAG estimator. 
According to Section \ref{meth}, we need an
initial covariance matrix estimation $\hat{\Sigma}_{init}$ in order to run
the PC-DAG technique. In Section \ref{meth}, we used the sample covariance
$\hat{\Sigma}_{init} = \hat{\Sigma}_{MLE}$ from (\ref{MLEcov}). It is well
known that the  standard sample covariance estimator is 
not robust against outliers or non-Gaussian distributions. 

In order to get a robust version of the PC-DAG
method we start with a robust  estimate of
$\Sigma$. We propose to use the orthogonalized Gnanadesikan-Kettenring (OGK)
estimator presented by \cite{OGK02}. Employing the OGK estimator in the
PC-algorithm, i.e. estimating partial correlations from the OGK covariance
estimate, we obtain a robustified estimate of the CPDAG, see also
\cite{ROBPC08}, and finally a robust PC-DAG covariance estimate as in
(\ref{estcov}) and  (\ref{DAGest-cov}) by using again the OGK covariance
estimator instead of $\hat{\Sigma}_{MLE}$. 

An ``ad-hoc'' robustification of 
the Glasso method can be achieved by using in (\ref{logliki}) the robust
OGK covariance estimate instead of the sample covariance
$\hat{\Sigma}_{MLE}$.   

\subsection{Simulation study for non-Gaussian data} 

In order to analyze the behavior of the robust PC-DAG method we use a
simulation model as in Section \ref{dagmodel} but with different
distributions for the errors $\epsilon$. Regarding the latter, we consider
the following distributions: $\mathcal{N}(0,1)$, 
$0.9\mathcal{N}(0,1)+0.1t_3(0,1)$ or
$0.9\mathcal{N}(0,1)+0.1\mbox{Cauchy}(0,1)$.

We compare the standard PC-DAG, robust PC-DAG, standard Glasso and the robust
Glasso estimators for Gaussian, 10$\%$ $t_3$ contaminated Gaussian and
10$\%$ Cauchy 
contaminated Gaussian data for one specific parameter setting:
\begin{itemize}

\item[R]: $n=50$, $p=80$, $s=0.01$
  
\end{itemize}
In order to compare the four methods we use the
Kullback-Leibler loss defined in Section \ref{dagmodel}.  
For the four estimation methods we plot the
Kullback-Leibler loss against the logarithm of the average number of
non-zero entries of the estimated concentration matrix $\hat{\Sigma}^{-1}$. 
The dotted vertical line represents the average
number of non-zero entries of the true underlying concentration
matrices. All the results are again based on 50 independent simulation runs.  

\begin{figure}[h]
  \centerline{ 
          \subfigure[Gaussian data]{\includegraphics[scale=0.35,width=6cm]{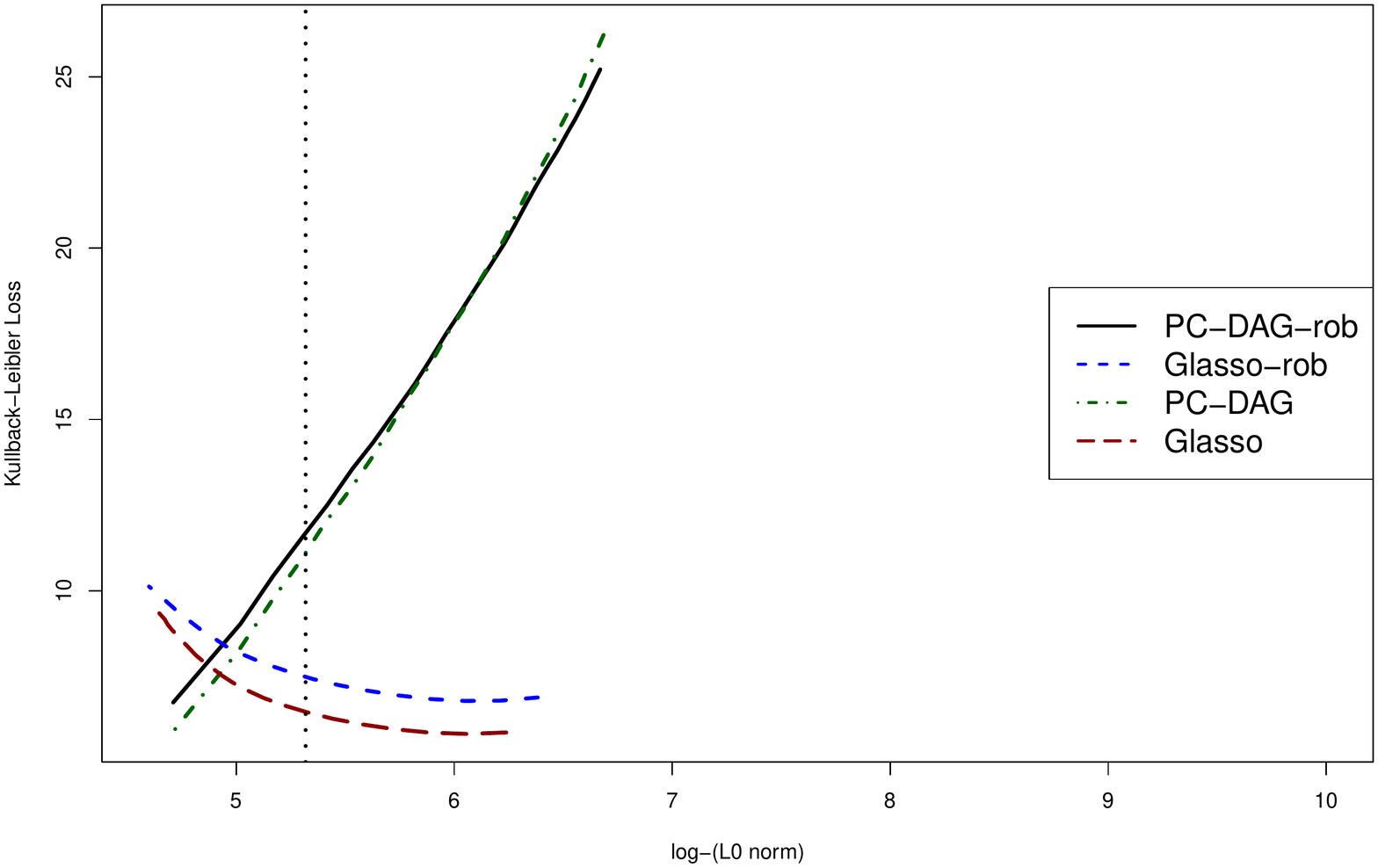}}
\subfigure[10$\%$ $t_3$ contaminated Gaussian data]{\includegraphics[scale=0.35,width=6cm]{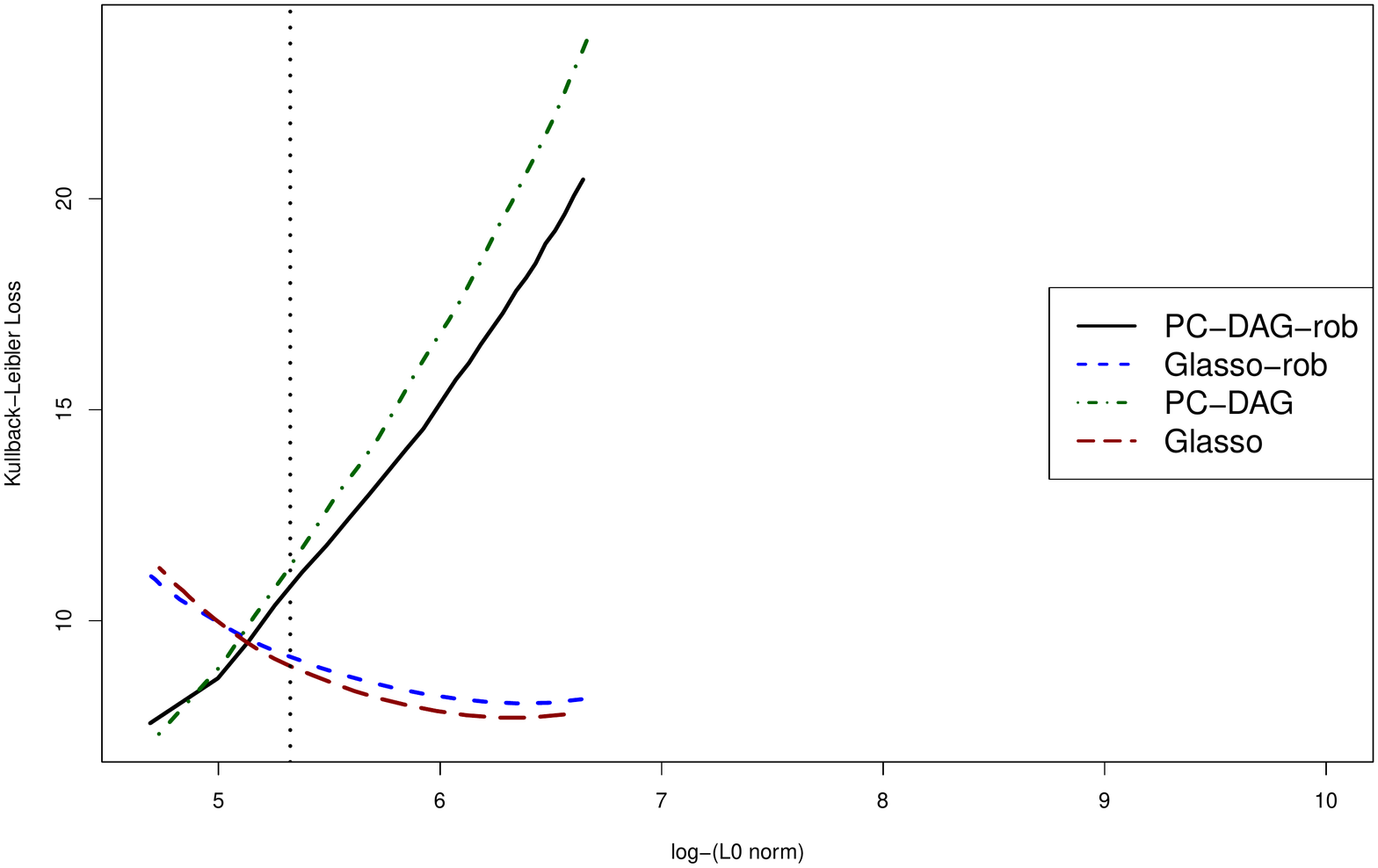}}} 
        \centerline{  
          \subfigure[10$\%$ Cauchy contaminated Gaussian data]{\includegraphics[scale=0.35,width=6cm]{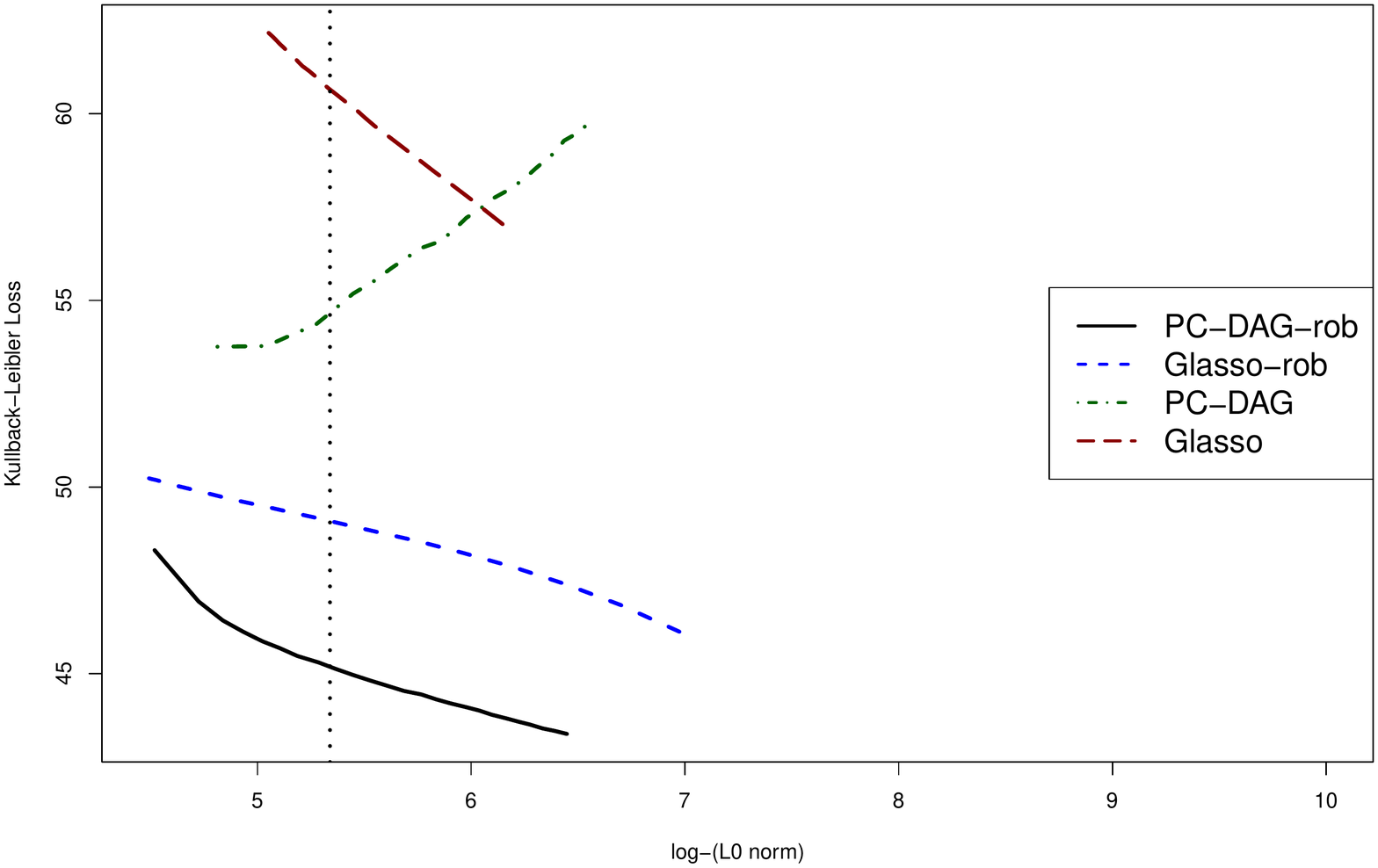}}}
        \caption{Kullback-Leibler loss
          against the logarithm of the average number of non-zero elements of
          $\Sigma^{-1}$ for Gaussian data 
          (a), 10$\%$ $t_3$ contaminated Gaussian data (b) and 10$\%$ Cauchy
          contaminated Gaussian data (c).} 
        \label{Rob}
      \end{figure}

Figures \ref{Rob} (a) and \ref{Rob} (b) show that without or with
moderate outliers, the standard and robust PC-DAG estimators perform about
as well as 
the standard and robust Glasso: the claim is based on the observation that
the minimum Kullback-Leibler loss of each of the four methods is about the
same, although the corresponding sparsity of the fitted concentration
matrix may be very different. In the presence of more severe outliers, the
robust PC-DAG technique is best as can be seen from Figure \ref{Rob} 
(c). In summary, the robust PC-DAG estimator is a
useful addition to gain robustness for estimating a high-dimensional
concentration matrix. 

\section{Summary and Discussion}

We have introduced the PC-DAG estimator, a graphical model based technique for
estimating sparse covariance 
matrices and their inverses from high-dimensional data. The method is based
on very different methodological concepts than shrinkage estimators. Our
PC-DAG procedure is  
invariant to variable permutation, yields a positive definite
estimate of the covariance and concentration matrix, and we have proven 
asymptotic consistency for sparse high-dimensional settings. An
implementation of the estimator is based on the \texttt{R}-package
\texttt{pcalg} \cite{pcalgMAN}. We remark that alternatively, one could
construct a high-dimensional covariance estimate based on a sparse
undirected conditional independence graph which itself can be inferred from
data using e.g. the node-wise Lasso procedure from \cite{MeinB06}.  

We have compared our PC-DAG estimator with the Glasso \cite{FHT07,BGA08} in two
simulation models. For the concentration matrix, our
PC-DAG approach clearly outperforms the Glasso technique for some parameter
settings, with performance gains up to 30-50\%, while it keeps up with
Glasso for the rest of the considered scenarios. For estimation of
covariances, the conclusions are similar but slightly less pronounced than
for inferring concentration matrices. Furthermore, we have compared
the two methods in two real data-sets and found only marginal differences in
performance. If the data generating mechanism is well
approximated by a DAG-model, the PC-DAG estimator is undoubtedly better
than the shrinkage-based Glasso. However, it is very hard to know a-priori
how well a DAG-model describes the underlying true distribution. 
Finally, we have presented a robustification of our PC-DAG estimator for
cases where the Gaussian data is contaminated by outliers. 

\section*{Appendix}
\begin{appendix}

\section{Proof of Lemma \ref{lemma:consistenz}}

A key element of the proofs is the analysis of low-order regression
problems described in Section \ref{estfromDAG}. 
%%%%%%%%%%%%%%%%%%%%%%%%%%%%%%%%%%%%%%%%%%%%%%%%%%%%%%%%%%%%%%%%%%%%%%%%%%%%%%%%%
For a DAG-structure with sets of parents, we consider regressions of the form
\begin{eqnarray*} 
X_i = \sum_{j \in pa(i)} \beta_j^{(i)} X_j + \eps_i,\ \ \eps_i \sim {\cal N}(0,\sigma^2_{i|pa(i)}),
\end{eqnarray*}
and $\eps_i$ independent of $X_{pa(i)}$. The corresponding OLS
estimates based on $n$ i.i.d. samples $X^{(1)},\ldots, X^{(n)}$ as in
(\ref{data}) are denoted by 
\begin{eqnarray*}
\hat{\beta}_j^{(i)},\ \ \hat{\sigma}^2_{i|pa(i)} = (n - |pa(i)|)^{-1} \sum_{r=1}^n
\left(X^{(r)}_{i} - \sum_{j \in pa(i)} \hat{\beta}_j^{(i)} X^{(r)}_{j}\right)^2.
\end{eqnarray*}

\begin{lemma}
\label{llb}
Suppose that the Gaussian assumption in (A), assumptions (B) and (F)
hold. Then, for every $\epsilon>0$, 
\begin{equation}
  \label{eq:blemma}
\begin{split}
  &\PR{\sup_{i=1,\ldots,p_n,j\in pa(i)}\left|
    \hat{\beta}^{(i)}_{j}-\beta^{(i)}_{j}\right|>\frac{\epsilon}{q_n}}\\ 
&\leq \frac{C_1}{\epsilon}q_n^2p_n\exp\left(-C_2\frac{\epsilon^2}{q_n^2}(n-q_n-1)\right)+2\exp\left(-C_3(\frac{n}{2}-q_n-1)\right),
\end{split}
\end{equation}
$n\geq 2(q_n + C_4)$ where $C_1,C_2>0$ are constants
depending on $\sigma^2$ and $r$ (see Assumptions (B) and (F)), and $C_3,
C_4>0$ are absolute constants. 
\end{lemma}

\begin{proof}[Proof]
 The proof is analogous to Lemma 7.1 in \cite{MMP09}. For completeness, we
 give a detailed derivation. The union bound yields
\begin{equation}
\label{lang}
  \PR{\sup_{i=1,\ldots,p_n,j\in pa(i)}\left|
    \hat{\beta}^{(i)}_{j}-\beta^{(i)}_{j}\right|>\frac{\epsilon}{q_n}} 
\leq p_nq_n\sup_{i,j}\PR{\left|
    \hat{\beta}^{(i)}_{j}-\beta^{(i)}_{j}\right|>\frac{\epsilon}{q_n}}.
\end{equation}
Next we analyze $\sup_{i,j}\PR{\left|
    \hat{\beta}^{(i)}_{j}-\beta^{(i)}_{j}\right|>\tilde{\epsilon}}$ for a
general $\tilde{\epsilon} >0$. 

Let $i \in \{1,...,p_n\}$ and denote by $s(i,j) = pa(i)\backslash j$. We
consider first the conditional distribution of
$\hat{\beta}^{(i)}_{j}|X_{pa(i)}$. The variance of $X_{i}|X_{pa(i)}$ is
$\sigma^2_{i|pa(i)}$ and we denote the variance of $X_{j}|X_{s(i,j)}$ by
$\sigma^2_{j|s(i,j)}$. Further, we denote the sample 
variance of $X_{j}$ by $\hat{\sigma}^2_{j}$, the sample variance of
$X_{j}|X_{s(i,j)}$ by $\hat{\sigma}^2_{j|s(i,j)}$ and the
sample multivariate correlation coefficient between $X_{j}$ and $X_{s(i,j)}$ by
$R^2_{j|s(i,j)}$. Then,  when conditioning on ${\bf X}_{pa(i)} =
\{X_{r,j};\ r=1,\ldots ,n,\ j \in pa(i)\}$, 
\begin{equation}
\label{var}
\VAR{\hat{\beta}^{(i)}_{j} \mid {\bf X}_{pa(i)}}=\frac{1}{1-R^2_{j|s(i,j)}}\frac{\sigma^2_{i|pa(i)}}{(n-1)\hat{\sigma}^2_{j}}=\frac{\sigma^2_{i|pa(i)}}{(n-|s(i,j)|-1)\hat{\sigma}^2_{j|s(i,j)}},
\end{equation}
where the first equality follows from e.g. \cite[p.120]{JF97} and the second
equality follows from
$1-R^2_{j|s(i,j)}=\frac{(n-|s(i,j)|-1)\hat{\sigma}^2_{j|s(i,j)}}{(n-1)\hat{\sigma}^2_{j}}$. With
(\ref{var}), $\ERW{\hat{\beta}^{(i)}_{j} \mid {\bf X}_{pa(i)}}=\beta^{(i)}_{j}$
and the Gaussian assumption in (A), we get
\begin{equation}
\label{prob}
\begin{split}
&\PR{|\hat{\beta}^{(i)}_{j}-\beta^{(i)}_{j}|>\tilde{\epsilon} \mid {\bf
    X}_{pa(i)}}=\\
&\PR{|Z|>\frac{\tilde{\epsilon}\sqrt{n-|s(i,j)|-1}
    \hat{\sigma}_{j|s(i,j)}} {\sigma_{i|pa(i)}} \mid {\bf X}_{pa(i)}},
\end{split}
\end{equation}
where $Z$ is a standard normal random variable.

We first analyze (\ref{prob}) on
the set $B_{js(i,j)}=\{\hat{\sigma}^2_{j|s(i,j)}>\frac{1}{2}\sigma^2_{j|s(i,j)}\}$.
From assumption (F) and $\VAR{X_{i}\mid X_{pa(i)}} \le \sigma^2$ it follows
that
\begin{eqnarray}\label{add1}
\inf_{i=1,...,p_n,j\in pa(i)}\frac{\VAR{X_{j}\mid  X_{pa(i)\setminus
      j}}}{\VAR{X_{i}\mid X_{pa(i)}}}\geq \frac{r}{\sigma^2}= v^2, 
\end{eqnarray}
where $v>0$. Using this bound from (\ref{add1}) we obtain

\begin{eqnarray}\label{add2}
&&\PR{|Z|>\frac{\tilde{\epsilon} \sqrt{n-|s(i,j)|-1}
    \hat{\sigma}_{j|s(i,j)}}{\sigma_{i|pa(i)}} \mid {\bf
    X}_{pa(i)}}\mathbb{I}_{B_{js(i,j)}}\nonumber \\ 
&&\leq  \PR{|Z|>\tilde{\epsilon} v
  \frac{\sqrt{n-|s(i,j)|-1}}{\sqrt{2}}}\nonumber \\
&&\leq \PR{|Z|>C \tilde{\epsilon} \sqrt{n-|q_n|-1}},
\end{eqnarray}
where C depends on $v$ in (\ref{add1}). We then bound the tail probability
of the standard normal distribution by 
$\PR{|Z|>a}\leq \frac{2}{\sqrt{2\pi}a}\exp{(\frac{-a^2}{2})}$ for $a >
0$. Hence, (\ref{add2}) can be further bounded by 
\begin{eqnarray}\label{add3}
\frac{C_1}{\tilde{\epsilon}}\exp{(-C_2\tilde{\epsilon}^2(n-q_n-1))}
\end{eqnarray}
for all $n$ such that $q_n<n-2$, where $C_1$, $C_2>0$ are constants
depending on $v$ in (\ref{add1}), i.e. they depend on $\sigma^2$ and $r$ in
assumptions (B) and (F). 

Next, we compute a bound for $\PR{B^C_{js(i,j)}}$. Note
that 
\[
\begin{split}
\PR{B^C_{js(i,j)}\mid X_{s(i,j)}}
&=
\PR{\frac{(n-|s(i,j)|-1)\hat{\sigma}^2_{j|s(i,j)}}{\sigma^2_{j|s(i,j)}}\leq
\frac{(n-|s(i,j)|-1)}{2} \mid X_{s(i,j)}}\\
&= \PR{\chi^2_{n-|s(i,j)|-1}\leq \frac{(n-|s(i,j)|-1)}{2}}\\
&\leq \PR{\chi^2_{n-q_n-1}\leq \frac{n-1}{2}}.
\end{split}
\]
Now we apply Bernstein's inequality \cite[Lemma 2.2.11]{vdVAW96} by
writing 
\[
\begin{split}
\PR{\chi^2_{n-q_n-1}\leq \frac{n-1}{2}}
&= \PR{\chi^2_{n-q_n-1}-(n-q_n-1)\leq \frac{-(n-1)}{2}+q_n}
\\
&\leq \PR{|\chi^2_{n-q_n-1}-(n-q_n-1)|< \frac{(n-1)}{2}-q_n}
\end{split}
\]
and noting that $\chi^2_{n-q_n-1}-(n-q_n-1)$ can be viewed as the sum
of $n-q_n-1$ independent centered $\chi^2_1$ random variables. Hence, the
last term is bounded above by 
\[
2\exp{\left(-\frac{(\frac{n-1}{2}-q_n)^2}{C^\prime_3+C^\prime_4(\frac{n-1}{2}-q_n)}\right)}
\]
where $C^\prime_3$, $C^\prime_4>0$ are constants arising from moment
conditions. This expression is in addition bounded above by
\begin{eqnarray}\label{add4}
2\exp{(-C_3(\frac{n}{2}-q_n-1))}
\end{eqnarray}
for all $n$ such that
$\frac{n-2}{2}-q_n>C^\prime_3$, and $C_3 > 0 $ is a constant arising from
moment conditions. Because this bound in (\ref{add4}) holds for all $X_{s(i,j)}$
with $|s(i,j)|\leq q_n$, it also holds for the unconditional probability
$\PR{B^C_{js(i,j)}}$.

The upper bound for
$\PR{|\hat{\beta}^{(i)}_j-\beta^{(i)}_j|>\tilde{\epsilon}}$ now follows
by combining (\ref{add3}) and (\ref{add4}): 
\[
\begin{split}
&\PR{\left|\hat{\beta}^{(i)}_j-\beta^{(i)}_j\right|>\tilde{\epsilon}}\\
&\leq \int_{B_{js(i,j)}}
\PR{\left|\hat{\beta}^{(i)}_j-\beta^{(i)}_j\right|>\tilde{\epsilon}\mid
  pa(i)}dF_{X_{j,s(i,j)}}+\PR{B^C_{js(i,j)}}
\\
&\leq \frac{C_1}{\tilde{\epsilon}}\exp{(-C_2\tilde{\epsilon}^2(n-q_n-1))}+2\exp{(-C_3(\frac{n}{2}-q_n-1))}.
\end{split}
\]

Now by using $\tilde{\epsilon}=\frac{\epsilon}{q_n}$ we derive
\begin{equation}
\label{supP}
\begin{split}
&\sup_{i,j}\PR{\left|\hat{\beta}^{(i)}_j-\beta^{(i)}_j\right|>\frac{\epsilon}{q_n}}\\
&\leq
\frac{C_1q_n}{\epsilon}\exp{(-C_2\frac{\epsilon^2}{q^2_n}(n-q_n-1))}+2\exp{(-C_3(\frac{n}{2}-q_n-1))}
\end{split}
\end{equation} 
which holds for all $n > 2 (q_n + C_3^\prime) + 2 = 2(q_n + C_4)$. 
Combining (\ref{supP}) with (\ref{lang}) we complete the proof of Lemma
\ref{llb}.  
 
\end{proof}

%%%%%%%%%%%%%%%%%%%%%%%%%%%%%%%%%%%%%%%%%%%%%%%%%%%%%%%%%%%%%%%%%%%%%%%%%%%%%%%%

\begin{lemma}
\label{llsigma}
Suppose that the Gaussian distribution in assumption (A), assumptions (B)
and (F) hold. Then, for every $\epsilon>0$,
\begin{eqnarray*}
&&\PR{\sup_{1\leq i \leq p_n}\left|
      \frac{1}{\hat{\sigma}^2_{i|pa(i)}}-\frac{1}{\sigma^2_{i|pa(i)}}\right|>\frac{\epsilon}{q_n}}\\
&&\leq  
  p_n 2 \left(\exp\left(-\frac{\epsilon^2(n-q_n)}{6C^2q_n^2\sigma^4+4C \epsilon
      q_n\sigma^2}\right) + \exp \left(-\frac{r^2(n - q_n)}{24 \sigma^4+ 8r
      \sigma^2}\right) \right)     
\end{eqnarray*}
where $C>0$ is an absolute constant and $r>0$ as in assumption (F).

\end{lemma}

\begin{proof}[Proof]

Using the union bound, for $\tilde\epsilon >0$,
\[
\PR{\sup_{i=1,\ldots ,p_n}\left|
    \hat{\sigma}^2_{i|pa(i)}-\sigma^2_{i|pa(i)}\right|>\tilde\epsilon} \leq
p_n \sup_{i=1,\ldots
  ,p_n}\PR{\left|\hat{\sigma}^2_{i|pa(i)}-\sigma^2_{i|pa(i)}\right|>\tilde\epsilon}.  
\]
For the conditional probability, when conditioning on ${\bf X}_{pa(i)} =
\{X_{r,j};\ r=1,\ldots ,n,\ j \in pa(i)\}$, we have that
$\PR{\left|\hat{\sigma}^2_{i|pa(i)}-\sigma^2_{i|pa(i)}\right|>\tilde{\epsilon} \mid
  {\bf X}_{pa(i)} }$ is equal to
\[
\begin{split}
  &\PR{\left|\frac{\hat{\sigma}^2_{i|pa(i)}}{\sigma^2_{i|pa(i)}}-1\right|>
    \frac{\tilde{\epsilon}}{\sigma^2_{i|pa(i)}}\mid {\bf X}_{pa(i)}}=\\  
 &\PR{\left|\frac{(n-|pa(i)|)\hat{\sigma}^2_{i|pa(i)}}{\sigma^2_{i|pa(i)}} -
     (n-|pa(i)|)\right|>\frac{\tilde{\epsilon}(n- |pa(i)|)}{\sigma^2_{i|pa(i)}}\mid
   {\bf X}_{pa(i)}}. 
\end{split}
\]
Because $\frac{(n-|pa(i)|)\hat{\sigma}^2_{i|pa(i)}}{\sigma^2_{i|pa(i)}}-(n-|pa(i)|)$ is a sum
of $(n-|pa(i)|)$ independent $\chi_1^2$-distributed centered random
variables, we can use Bernstein's inequality \cite[Lemma
2.2.11]{vdVAW96}. Hence, with $\sigma^2_{i|pa(i)} \leq \sigma^2$ we get   
\begin{eqnarray*}
&&\PR{\left|\frac{(n-|pa(i)|)\hat{\sigma}^2_{i|pa(i)}}{\sigma^2_{i|pa(i)}} -
     (n-|pa(i)|)\right|>\frac{\tilde\epsilon(n-|pa(i)|)}{\sigma^2_k} \mid
   {\bf X}_{pa(i)}}\\ 
&&\leq 2\exp\left(-\frac{\tilde\epsilon^2(n-|pa(i)|)}{6\sigma^4+4\tilde\epsilon
     \sigma^2}\right).    
\end{eqnarray*}
Since this bound holds for all ${\bf X}_{pa(i)}$, the bound also applies 
to the unconditional probability:
\begin{eqnarray}\label{boundlemma}
&&\PR{\left|\hat{\sigma}^2_{i|pa(i)}-\sigma^2_{i|pa(i)}\right|>\tilde{\epsilon}}\nonumber\\ 
&&=\PR{\left|\frac{(n-|pa(i)|)\hat{\sigma}^2_{i|pa(i)}}{\sigma^2_{i|pa(i)}} -
     (n-|pa(i)|)\right|>\frac{\tilde\epsilon(n-|pa(i)|)}{\sigma^2_k}
 }\nonumber\\
&&\leq 2\exp\left(-\frac{\tilde\epsilon^2(n-|pa(i)|)}{6\sigma^4+4\tilde\epsilon
     \sigma^2}\right).    
\end{eqnarray}

We use now a Taylor expansion:
\[
\frac{1}{\hat{\sigma}^2_{i|pa(i)}}=\frac{1}{\sigma^2_{i|pa(i)}}-\frac{1}{\tilde{\sigma}^4_{i|pa(i)}}(\hat{\sigma}^2_{i|pa(i)}-\sigma^2_{i|pa(i)}),
\]
where $\left|\tilde{\sigma}^2_{i|pa(i)}-\sigma^2_{i|pa(i)}\right|\leq
\left|\hat{\sigma}^2_{i|pa(i)}-\sigma^2_{i|pa(i)}\right|$.\\ Consider the set $B =
\{\sup_{i=1,\ldots ,p_n} \left|\hat{\sigma}^2_{i|pa(i)}-\sigma^2_{i|pa(i)}\right| \le r/2\}$
with $r>0$ as in assumption (F). Then, on $B$, we have
$\left|\frac{1}{\tilde{\sigma}^4_{i}}\right|\leq \tilde{C} < \infty$ (and
the bound does not depend on the index $i$). Therefore,
\[
\begin{split}
  &\PR{\sup_{i=1,\ldots ,p_n}\left|
    \frac{1}{\hat{\sigma}^2_{i|pa(i)}}-\frac{1}{\sigma^2_{i|pa(i)}}\right|>\frac{\epsilon}{q_n}}\\ 
&\leq
\PR{\left\{\tilde{C}\sup_{i=1,\ldots ,p_n}
  \left|\hat{\sigma}^2_{i|pa(i)}-\sigma^2_{i|pa(i)}\right|>\frac{\epsilon}{q_n}\right\} \cap B} + \PR{B^C}
\end{split}
\]
The first term and second term on the right-hand side can be bounded using
(\ref{boundlemma}), leading to the bound in the statement of the lemma. This
completes the proof of Lemma \ref{llsigma}.

\end{proof}

%%%%%%%%%%%%%%%%%%%%%%%%%%%%%%%%%%%%%%%%%%%%%%%%%%%%%%%%%%%%%%%%%%%%%%%%%%%%%%%%% 

\begin{proof}[Proof of Lemma \ref{lemma:consistenz}]

Let $G$ be a DAG from the true underlying CPDAG, i.e. the true equivalence
class. Using the union bound we have
\begin{eqnarray}\label{add5}
\PR{\sup_{i,j=1,\ldots ,p_n}\left|
    \hat{\Sigma}^{-1}_{G,n;i,j}-\Sigma^{-1}_{n;i,j}\right|>\gamma} \leq
p_n^2\sup_{i,j}\PR{\left|\hat{\Sigma}^{-1}_{G,n;i,j}-\Sigma^{-1}_{n;i,j}\right|
  >\gamma}.   
\end{eqnarray}
Since $\hat{\Sigma}^{-1}=\hat{A}^{T}\COVH{\epsilon}^{-1}\hat{A}$ we have
$\hat{\Sigma}^{-1}_{G,n;i,j}=\sum_{k=1}^{p_n}\hat{\lambda}_k 
\hat{A}_{kj}\hat{A}_{ki}$ with
$\hat{\lambda}_k=\frac{1}{\hat{\sigma}^2_k}$ and $\hat{A}$ as in
(\ref{estcov}). Thus, 
\[
\begin{split}
  \left|\hat{\Sigma}^{-1}_{G,n;i,j}-\Sigma^{-1}_{n;i,j}\right| &=
  \left|\sum_{k=1}^{p_n}\left(\hat{\lambda}_k\hat{A}_{kj}\hat{A}_{ki}-\lambda_k 
      A_{kj}A_{ki}\right)\right|\\
&\leq \sum_{k=1}^{p_n}\left|\hat{\lambda}_k\hat{A}_{kj}\hat{A}_{ki}-\lambda_k
  A_{kj}A_{ki}\right|\\ 
&= \sum_{k=1}^{p_n}\left|\hat{\lambda}_k\hat{A}_{kj}\hat{A}_{ki}-\hat{\lambda}_kA_{kj}A_{ki}+\hat{\lambda}_kA_{kj}A_{ki}-\lambda_k
  A_{kj}A_{ki}\right|\\ 
&= \sum_{k=1}^{p_n}\left|\hat{\lambda}_k\left(\hat{A}_{kj}\hat{A}_{ki}-
    A_{kj}A_{ki}\right)+A_{kj}A_{ki}\left(\hat{\lambda}_k-\lambda_k\right)\right|\\
&\leq \sum_{k=1}^{p_n}\left(\left|\hat{\lambda}_k\right|\left|\hat{A}_{kj}\hat{A}_{ki}- A_{kj}A_{ki}\right|+\left|A_{kj}A_{ki}\right|\left|\hat{\lambda}_k-\lambda_k\right|\right)
\end{split}
\]
Consider the terms $\left|\hat{A}_{kj}\hat{A}_{ki}-
  A_{kj}A_{ki}\right|$ and $\left|\hat{\lambda}_k\right|$:
\[
\begin{split}
  \left|\hat{A}_{kj}\hat{A}_{ki}- A_{kj}A_{ki}\right|
&= \left|\hat{A}_{kj}\hat{A}_{ki}-\hat{A}_{kj}A_{ki}+\hat{A}_{kj}A_{ki}-
  A_{kj}A_{ki}\right|\\ 
&= \left|\hat{A}_{kj}\left(\hat{A}_{ki}-A_{ki}\right)+A_{ki}\left(\hat{A}_{kj}-
  A_{kj}\right)\right|\\ 
&\leq \left|\hat{A}_{kj}\right|\left|\hat{A}_{ki}-A_{ki}\right|+\left|A_{ki}\right|\left|\hat{A}_{kj}-A_{kj}\right|\\ 
\end{split}
\]
\[
\begin{split}
  \left|\hat{\lambda}_k\right|
= \left|\hat{\lambda}_k-\lambda_k+\lambda_k\right| \leq \left|\hat{\lambda}_k-\lambda_k\right|+\left|\lambda_k\right| 
\end{split}
\]
By plugging these bounds into the formula above and using that the
summations are over at most $q_n$ terms only (due to sparsity of
$\hat{A}_{ki}$ and $A_{ki}$), we obtain 
\[
\left|\hat{\Sigma}^{-1}_{G,n;i,j}-\Sigma^{-1}_{n;i,j}\right| \leq C q_n \delta
\] 
where $C >0$ is an absolute constant and $\delta$ the maximal absolute
difference of $\hat{A}$'s and $\hat{\lambda}$'s:
\begin{eqnarray*}
\delta = \max\{\max_{i,k}|\hat{A}_{ki} - A_{ki}|,\max_k |\hat{\lambda}_k -
\lambda_k|\}. 
\end{eqnarray*}
Hence
\[
\PR{\left|\hat{\Sigma}^{-1}_{G,n;i,j}-\Sigma^{-1}_{n;i,j}\right|>\gamma} \leq
\PR{C q_n\left|\delta\right|>\gamma}
=\PR{\left|\delta\right|>\frac{\gamma}{Cq_n}}
=\PR{\left|\delta\right|>\frac{\epsilon}{q_n}} 
\]
with $\frac{\gamma}{C}=\epsilon$. Because the convergence of the term
$\PR{\left|\delta\right|>\frac{\epsilon}{q_n}}$ is covered either by Lemma
\ref{llb} or Lemma \ref{llsigma}, since $q_n^2 = O(n^{1 - 2b})\ (0<b\le
1/2)$ from assumption (D), and using (\ref{add5}), we complete the proof of
Lemma \ref{lemma:consistenz}. 
\end{proof}

\section{Modifications of the PC-DAG covariance estimator}\label{sec.modif}  

With finite sample size, the PC-algorithm may make some errors. One of them
can produce conflicting v-structures when orienting the graph: if so, we
deal with it by keeping one and discarding other v-structures. In our
implementation, the result then depends on  
the order of the performed independence tests. Furthermore, it may happen
that the output of the PC-algorithm is an invalid CPDAG which does not
describe an equivalence class of DAGs. In such a case we use the
\emph{retry} type orientation procedure implemented in the 
\texttt{pcAlgo}-function of the \texttt{pcalg}-package, see the reference
manual of the \texttt{pcalg}-package \cite{pcalgMAN} for more information.  
\end{appendix}

\bibliography{covDAG}

\end{document}